\documentclass[a4paper,10pt,numbers=noenddot,twoside=on,headinclude=true,footinclude=false,headsepline=true]{scrartcl}

\usepackage[utf8]{inputenc}
\usepackage[english]{babel}
\usepackage{color}
\usepackage{amssymb}
\usepackage{amsmath}
\usepackage{amsfonts}
\usepackage{amsopn}
\usepackage{amsbsy}
\usepackage{enumerate}
\usepackage[plainpages=false,pdfpagelabels,breaklinks=true]{hyperref}
\usepackage{graphics}

\usepackage[amsmath,thmmarks,thref,hyperref]{ntheorem}

\theoremstyle{plain}
\newtheorem{theorem}{Theorem}[section]

\newtheorem{lemma}[theorem]{Lemma}
\newtheorem{corollary}[theorem]{Corollary}
\newtheorem{proposition}[theorem]{Proposition}
\newtheorem{assumption}[theorem]{Assumption}

\theorembodyfont{\upshape}
\newtheorem{remark}[theorem]{Remark}

\theoremstyle{nonumberplain}

\theoremsymbol{\ensuremath{_\Box}}
\newtheorem{proof}{Proof}

\numberwithin{equation}{section} 

\usepackage[scaled]{beramono}
\usepackage{helvet}
\usepackage[charter]{mathdesign} 
\SetMathAlphabet{\mathcal}{normal}{OMS}{cmsy}{m}{n} 
\SetMathAlphabet{\mathcal}{bold}{OMS}{cmsy}{m}{n} 
\providecommand{\ie}{i.~e.~}
\providecommand{\eg}{e.~g.~}
\providecommand{\cf}{cf.~}

\providecommand{\R}{\mathbb{R}}

\providecommand{\C}{\mathbb{C}}
\renewcommand{\C}{\mathbb{C}}

\providecommand{\T}{\mathbb{T}}
\renewcommand{\T}{\mathbb{T}}
\providecommand{\N}{\mathbb{N}}
\providecommand{\Z}{\mathbb{Z}}
\providecommand{\ii}{\mathrm{i}}
\providecommand{\e}{\mathrm{e}}
\renewcommand{\Re}{\mathrm{Re} \,}
\renewcommand{\Im}{\mathrm{Im} \,}

\providecommand{\eps}{\varepsilon}
\providecommand{\Cont}{\mathcal{C}}

\providecommand{\ran}{\mathrm{ran} \, }

\providecommand{\trace}{\mathrm{Tr} \,}

\providecommand{\dd}{\mathrm{d}}
\providecommand{\id}{\mathbf{1}}

\providecommand{\order}{\mathcal{O}}
\providecommand{\Fourier}{\mathcal{F}}

\providecommand{\trace}{\mathrm{Tr}}

\providecommand{\abs}[1]{\left \lvert #1 \right \rvert}
\providecommand{\sabs}[1]{\lvert #1 \vert}
\providecommand{\babs}[1]{\bigl \lvert #1 \bigr \rvert}

\providecommand{\norm}[1]{\left \lVert #1 \right \rVert}
\providecommand{\snorm}[1]{\lVert #1 \rVert}
\providecommand{\bnorm}[1]{\bigl \lVert #1 \bigr \rVert}

\providecommand{\bscpro}[2]{\bigl \langle #1 , #2 \bigr \rangle}

\providecommand{\sopro}[2]{\vert #1 \rangle \langle #2 \vert}

\providecommand{\ncint}{\mathrel{{\ooalign{$\int$\cr\kern+.07em\raise.15ex\hbox{$\pmb{\scriptstyle-}$}\cr}}}} \providecommand{\ncpartial}{\mathrel{{\ooalign{$\partial$\cr\kern+.29em\raise.79ex\hbox{$\pmb{\scriptstyle-}$}\cr}}}}

\providecommand{\Qgap}{Q_{\EFermi}}

\providecommand{\Alg}{\mathcal{A}}
\providecommand{\tracevol}{\mathcal{T}}
\providecommand{\Pol}{\mathcal{P}}
\providecommand{\shift}{\mathfrak{s}}

\providecommand{\Polad}{\Pol}
\providecommand{\Sone}{\mathbb{S}^1}
\providecommand{\Stwo}{\mathbb{S}^2}
\providecommand{\Blochb}{\mathcal{E}_{\mathrm{B}}}
\providecommand{\sphereb}{\mathcal{E}_{\Stwo}}
\providecommand{\Nvar}{\mathcal{N}}
\providecommand{\Svar}{\mathcal{S}}
\providecommand{\EFermi}{E}

\usepackage[style=alphabetic,
	citestyle=alphabetic,
	firstinits=true,
	backend=biber,
	hyperref=auto,
	sorting=nyt,
	maxcitenames=3,
	maxbibnames=99,
	minnames=3,
	arxiv=pdf,
	url=false,
	isbn=false,
	dateabbrev=true,
	datezeros=true,
	bibencoding=utf8]{biblatex}
\newbibmacro{string+doi}[1]{%
	\iffieldundef{doi}{#1}{\href{http://dx.doi.org/\thefield{doi}}{#1}}}
\DeclareFieldFormat
	{title}{\usebibmacro{string+doi}{\mkbibemph{#1}}\isdot}
\DeclareFieldFormat
	[article,inbook,incollection,inproceedings,patent,thesis,unpublished]
	{title}{\usebibmacro{string+doi}{\mkbibemph{#1}}\isdot}
\DeclareFieldFormat
	[article,inbook,incollection,inproceedings,patent,thesis,unpublished]
	{pages}{#1}
\DeclareFieldFormat
	[article]
	{journaltitle}{#1\isdot}
\DeclareFieldFormat
	[article]
	{volume}{\textbf{#1}}
\DefineBibliographyStrings{english}{pages={}}
\DeclareBibliographyDriver{article}{%
	\usebibmacro{bibindex}%
	\usebibmacro{begentry}%
	\usebibmacro{author/translator+others}%
	\setunit{\labelnamepunct}\newblock
	\usebibmacro{title}%
	\usebibmacro{journal+issuetitle}%
	\setunit*{\addcomma\space}
	\printfield{pages}
	\setunit*{\addcomma\space}
	\printfield{year}
	\usebibmacro{finentry}%
}

\newbibmacro*{journal+issuetitle}{%
	\usebibmacro{journal}%
	\setunit*{\addspace}%
	\iffieldundef{series}
	{}
	{\newunit
		\printfield{series}%
		\setunit{\addspace}}%
	\iffieldundef{volume}
		{}
		{\printfield{volume}%
		\setunit{\addspace}}%
	\setunit*{\addcolon\space}%
	\usebibmacro{issue}%
	\newunit}
\bibliography{bibliography}

\usepackage{units}
\usepackage{diagxy}
\usepackage{cancel}
\usepackage{arydshln}
\usepackage{dsfont}

\title{Topological Polarization \\ in Graphene-like Systems}
\author{Giuseppe De Nittis${}^{\ast}$ \& Max Lein${}^{\star}$}

\begin{document}

\maketitle
\vspace{-9mm}
\begin{center}
	$^{\ast}$ Department Mathematik, Universität Erlangen-Nürnberg \linebreak
	Cauerstrasse 11, D-91058 Erlangen, Germany \linebreak
	{\footnotesize \href{mailto:denittis@math.fau.de}{\texttt{denittis@math.fau.de}}}
	\medskip
	\\
	$^{\star}$ Kyushu University, Faculty of Mathematics \linebreak
	744 Motooka, Nishiku, Fukuoka, 819-0395, Japan \linebreak
	{\footnotesize \href{mailto:lein@ma.tum.de}{\texttt{lein@ma.tum.de}}}
\end{center}
\begin{abstract}
	In this article we investigate the possibility of generating \emph{piezoelectric orbital polarization} in graphene-like systems which are deformed periodically. 
	We start with discrete two-band models which depend on control parameters; in this setting, time-dependent model hamiltonians are described by loops in parameter space. Then, the gap structure at a given Fermi energy generates a non-trivial topology on parameter space which then leads to possibly non-trivial polarizations. 
	More precisely, we show the polarization, as given by the \emph{King-Smith--Vanderbilt formula}, depends only on the homotopy class of the loop; hence, a necessary condition for non-trivial piezo effects is that the fundamental group of the gapped parameter space must not be trivial. 
	The use of the framework of non-commutative geometry implies our results extend to systems with \emph{weak} disorder. 
	We then apply this analysis to the \emph{uniaxial strain model} for graphene which includes nearest-neighbor hopping and a stagger potential, and show that it supports non-trivial piezo effects; this is in agreement with recent physics literature. 
\end{abstract}

\tableofcontents

\section{Introduction and main results} 
\label{intro}
Piezoelectric materials are crystalline solids which become macroscopically charged when subjected to mechanical strain. One material that has recently moved into the limelight of piezoelectric physics due to the theoretical work \cite{Ong_Reed:engineered_piezoelectricity_graphene:2012} is graphene, and an experimental realization of these ideas would open up a lot of possibilities in the engineering of new piezoelectric devices. Much of graphene's peculiar properties \cite{Castro_Neto_et_al:electronic_properties_graphene:2009} stem from the conical intersections of valence and conduction band right at the Fermi energy. But the reason why it is an interesting material for piezoelectric devices is its unique mechanical robustness, allowing elastic deformations of up to $\unit[20]{\%}$ (as opposed to $\leqslant \unit[0.1]{\%}$ for normal materials) \cite{Liu_Ming_Li:2007,Lee_Wei_Kysar_Hone:elastic_properties_graphene:2008,Kim_et.all:2009}. 

To understand the link between graphene's band structure and piezoelectric properties, one needs a microscopic description of the piezoelectric effect. Such a description had eluded theoretical physicists until the mid-1970s when Martin \cite{Martin:polarization:1974} noticed that previous definitions of polarization in terms of microscopic dipole moments were incomplete. It took another 20 years until Resta \cite{Resta:electric_polarization:1992} and King-Smith and Vanderbilt \cite{King-Smith_Vanderbilt:polarization:1993} derived a formula for polarization from linear response theory. They recognized the crucial role of the adiabatic \emph{Berry phase} \cite{Berry:adiabatic:1984,Xiao_Chang_Niu:Berry_phase_electronic_properties:2010} and linked the difference in charge, the \emph{polarization $\Delta \mathcal{P} = \bigl ( \Delta \mathcal{P}_1 , \ldots , \Delta \mathcal{P}_d \bigr )$}, accumulated during a deformation in the time interval $[0,T]$ to 
\begin{align}
	\Delta \mathcal{P}_j := \ii \, \int_0^T \dd t \; \tracevol \Bigl ( P(t) \; \bigl [ \partial_t P(t) \; , \; \nabla_j P(t) \bigr ] \Bigr ) 
	.
	\label{intro:eqn:KSV_formula}
\end{align}
Here, $\tracevol$ denotes the \emph{trace per unit volume}, $P(t) = 1_{(-\infty,\EFermi)} \bigl ( H(t) \bigr )$ is the projection onto all states below the Fermi energy $\EFermi$ and $H(t)$ is the hamiltonian of the system. This equation is structurally identical to that for computing other Chern numbers such as those for the quantum Hall effect \cite{Thouless_Kohmoto_Nightingale_Den_Nijs:quantized_hall_conductance:1982,Bellissard_van_Elst_Schulz_Baldes:noncommutative_geometry_quantum_hall_effect:1994}. One crucial ingredient for the piezoelectric effect to occur is the absence of conducting states around the Fermi energy, \eg materials with a \emph{spectral gap} are good candidates. 

A mathematical justification of \eqref{intro:eqn:KSV_formula}, also called \emph{King-Smith--Vanderbilt formula}, has first been achieved by Panati, Sparber and Teufel \cite{Panati_Sparber_Teufel:polarization:2006} for the (commutative) case of continuous Schrödinger operators. In a later work, Schulz-Baldes and Teufel \cite{Schulz-Baldes_Teufel:random_polarization:2012} used the language of non-commutative geometry to establish \eqref{intro:eqn:KSV_formula} for \emph{dirty lattice systems}, \ie discrete operators which include the effects of random impurities. Both works also explore the topological nature of $\Delta \mathcal{P}$ in the case of \emph{periodic deformations} where $\Delta \mathcal{P}$ is quantized in appropriate units\footnote{The units for polarization are $\mbox{charge density} \times \mbox{distance}$, and a more careful consideration after restoring physical units yields $\Delta \Pol = \frac{\e}{\sabs{\mathbb{V}}} \, \sum_{j = 1}^d \Delta \Pol_j \, \gamma_j$ where $\e$ is the electron charge, $\sabs{\mathbb{V}}$ the volume of the Wigner-Seitz cell and the $\gamma_j$ are a basis for the lattice $\Gamma$ \cite[equation~(13)]{King-Smith_Vanderbilt:polarization:1993}. }; what is missing, however, are criteria that tell us \emph{which} periodic deformations lead to non-trivial polarization $\Delta \mathcal{P} \neq 0$. 

The main focus of this paper is the study of the piezoelectric effect for graphene subjected to periodic deformations. We will study the simplest kind of tight-binding model, the so-called \emph{uniaxial strain model}. It includes only nearest-neighbor interactions and a \emph{stagger potential}, and will be explained in more detail below. Our investigation has led us to the following three questions: 
\begin{itemize}
	\item[(Q1)] What is the topological origin of non-trivial polarizations?
	\item[(Q2)] Are there sufficiently general models applicable to graphene for which $\Delta \mathcal{P} \neq 0$? 
	\item[(Q3)] Is $\Delta \mathcal{P}$ stable under perturbations? 
\end{itemize}
Because our ideas can in principle be applied to any parameter-dependent system, we will formulate the first part of this work in more generality.

\subsection{The topology of the parameter space} 
\label{intro:parameter_space}
For the models we study, the hamiltonian $H(q)$ is described by a set of \emph{control parameters} $q = (q_1 , \ldots , q_N) \in \R^N$, \ie 
\begin{align*}
	H : Q \longrightarrow \Alg
\end{align*}
is a continuous (or even more regular) function that takes values in the selfadjoint elements of some algebra of bounded operators $\Alg$. We will always assume that $Q$ is a subset of $\R^N$. These parameters $q$ model the influence of external effects on the quantum system; In our example, the $q_j$ could be hopping parameters. 

Any choice of Fermi energy $\EFermi$ singles out configurations denoted with $Q_{\EFermi}$ made up of those values of $q \in Q$ for which 
\begin{enumerate}[(i)]
	\item $\EFermi$ lies in a spectral gap of $H(q)$, $\EFermi \not\in \sigma \bigl ( H(q) \bigr )$, and 
	\item there are states below $\EFermi$, \ie $P(q) := 1_{(-\infty,\EFermi)} \bigl ( H(q) \bigr ) \neq 0$. 
\end{enumerate}
We shall refer to $Q_{\EFermi}$ as the \emph{space of gapped configurations at $\EFermi$}. It is the topology of $Q_{\EFermi}$ which determines whether $\Delta \mathcal{P} = 0$ or not; more precisely, the \emph{fundamental group} $\pi_1(Q_{\EFermi})$ \cite{Hatcher:algebraic_topology:2002} can provide a classification for the piezoelectric effects for a given model system at a given Fermi energy. The idea is as follows: To each physical deformation which does not close the gap at $\EFermi$, we can associate a loop in parameter space $\eta : \Sone \rightarrow Q_{\EFermi}$ so that the time-dependent, $T$-periodic hamiltonian 
\begin{align}
	H_{\eta}(t) = H \bigl ( \eta(\nicefrac{2\pi t}{T}) \bigr ) 
	. 
	\label{intro:eqn:relation_model_hamiltonian_loops}
\end{align}
can be expressed in terms of the loop $\eta$ and the model hamiltonian $H$. The fact that the deformation should be continuous implies that $\eta$ is continuous. Then $H_{\eta}$ in turn defines a time-dependent Fermi projection 
\begin{align}
	P_{\eta}(t) = 1_{(-\infty,\EFermi)} \bigl ( H_{\eta}(t) \bigr ) 
	\label{intro:eqn:Fermi_projection}
\end{align}
which is then plugged into equation~\eqref{intro:eqn:KSV_formula} to obtain $\Delta \mathcal{P}(\eta)$ for each loop $\eta$. Overall, this procedure yields a map 
\begin{align*}
	\eta \mapsto \Delta \mathcal{P}(\eta) \in \Z^d
\end{align*}
from the space of loops in $Q_{\EFermi}$. 

That $\Delta \mathcal{P}(\eta)$ is a topological quantity is reflected in the fact that the value depends only on the equivalence class $[\eta] \in \pi_1(Q_{\EFermi})$: if $\eta$ and $\eta'$ are homotopic, then also the corresponding Fermi projections $P_{\eta}$ and $P_{\eta'}$ can be continuously deformed into one another (Proposition~\ref{KSV:lem:homotopic_loops_homotopic_projections}). Thus, the invariance of \eqref{intro:eqn:KSV_formula} under homotopies yields 
\begin{theorem}[Homotopy-invariance of $\Delta \mathcal{P}$]\label{intro:thm:homotopy_invariance_polarization}
	Under the technical conditions enumerated in Theorem~\ref{main_result:thm:Delta_Pol_Poincare_group}, the map 
	\begin{align}
		\Delta \mathcal{P}_{\ast} : \pi_1(Q_{\EFermi}) \longrightarrow \Z^d 
		, \quad 
		[\eta] \mapsto \Delta \mathcal{P}_{\ast}([\eta]) := \Delta \mathcal{P}(\eta)
		, 
		\label{intro:eqn:polarization_morphism}
	\end{align}
	is a group morphism where $\eta$ is any continuously differentiable representative of $[\eta]$. 
\end{theorem}
The practical implication of this theorem for calculations is as follows: given a model $H : Q \longrightarrow \Alg$, a value for the Fermi energy $\EFermi$ and a deformation $\eta$, all we need to figure out is the \emph{equivalence class} $[\eta] \in \pi_1(Q_{\EFermi})$. Then we are free to use \emph{any} loop $\eta'$ in the equivalence class $[\eta]$ to compute the polarization. Moreover, we get an immediate criterion for the triviality of the polarization: 
\begin{corollary}\label{intro:cor:criterion_triviality}
	A necessary condition for $\Delta \mathcal{P}(\eta) \neq 0$ is $\pi_1(Q_{\EFermi}) \neq \{ 0 \}$. 
\end{corollary}
If $\pi_1(Q_{\EFermi}) = \{ 0 \}$, then all deformations supported in $Q_{\EFermi}$ are equivalent to the case of no deformation, and hence $\Delta \mathcal{P}(\eta) = 0$. It is in this sense that regions in $Q$ where the spectral gap at $\EFermi$ closes create the non-trivial topology necessary for a non-trivial piezo effect. 
\medskip

\noindent
The linearity of $\Delta \Pol_{\ast}$ implies it sends commutators of loops to $0$, and thus \eqref{intro:eqn:polarization_morphism} is also well-defined as a map from the \emph{abelianization} of $\pi_1(\Qgap)$ to $\Z^d$. The abelianization, however, corresponds to the homology group $H_1(\Qgap)$ \cite{Hatcher:algebraic_topology:2002}. This distinction, which is only important in the `exotic' cases of non-commutative $\pi_1(\Qgap)$, shows that the piezoelectricity depends more properly on the homological properties of $\Qgap$. 

\subsection{Tight-binding models for piezoelectricity in graphene} 
\label{intro:tight_binding_graphene}
After answering (Q1), let us turn our attention to graphene and the second question. To the best of our knowledge, with the exception of the Rice-Mele model in $d = 1$ \cite{Rice_Mele:1982,Onoda_Murakami_Nagaosa:2004}, there are no other concrete models in $d > 1$ for which the polarization has been calculated exactly. Our framework allows for the evaluation of Chern numbers (including the polarization) for so-called \emph{two-band hamiltonians} (\cf Section~\ref{topology}); in particular, we will consider the \emph{uniaxial strain model} in $d = 2$. This simple model incorporates all the hallmarks of piezoelectrical modifications of graphene proposed by theoretical physicists \cite{Ong_Reed:engineered_piezoelectricity_graphene:2012}, and we show that it allows for deformations which have non-trivial polarizations. 

Let us quickly recount the basics of the crystal structure of graphene: it is an essentially two-dimensional material consisting of a single layer of graphite. The carbon atoms are arranged in a honeycomb lattice (\cf Figure~\ref{intro:figure:honeycomb}) which is obtained by the juxtaposition of two triangular lattices $\Gamma \simeq \Z^2$ generated by the fundamental vectors 
\begin{align*}
	\gamma_1 = \tfrac{a}{2} \bigl ( 3 , + \sqrt{3} \bigr ) 
	, 
	&&
	\gamma_2 = \tfrac{a}{2} \bigl ( 3 , - \sqrt{3} \bigr ) 
	. 
\end{align*}
\begin{figure}[t]
	\hfil
	\resizebox{!}{80mm}{
		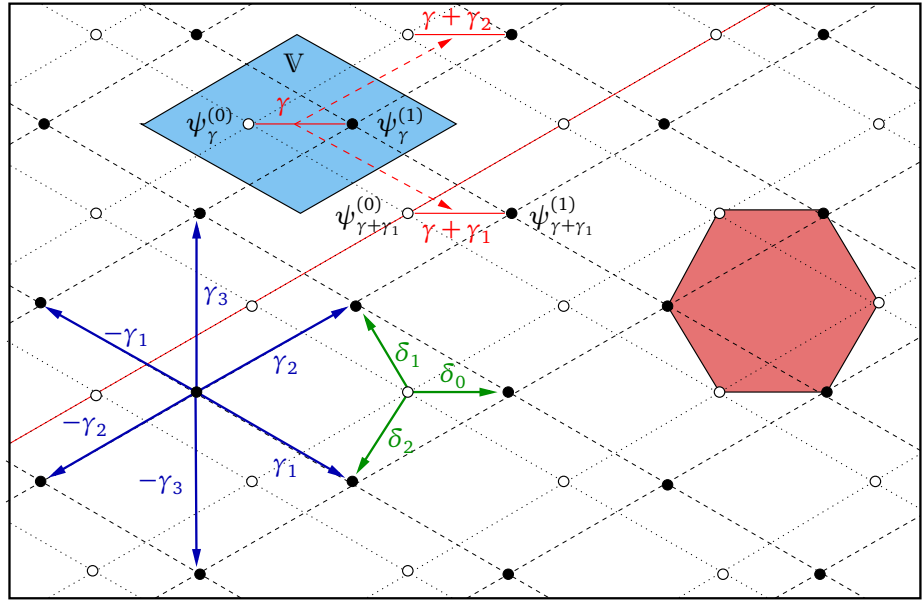
	}
	\hfil
	\caption{The honeycomb lattice as the superposition of two triangular lattices $\Gamma$ with atoms of type $0$ (white) and type $1$ (black). Every
\emph{unit cell} $\mathbb{V}$ contains a pair of sites one of type 0 and the other of type 1. The green arrows $\delta_0,\delta_1,\delta_3$ connect sites of type 0 with the three \emph{nearest-neighbor sites} of type 1.
The blue arrows $\pm\gamma_1,\pm\gamma_2,\pm\gamma_3$ connect a given site with the six \emph{second-nearest-neighbor} sites.
$\psi^{(0)}_\gamma$ (resp. $\psi^{(1)}_\gamma$) denotes the component of the wave function for sites of type 0 (resp. 1) in the cell located in 
$\gamma$.}
	\label{intro:figure:honeycomb}
\end{figure}
Here $a \approx \unit[1.42]{\,}$Å is the distance between two carbon atoms. The vectors 
\begin{align*}
	\delta_0 = a (1,0)
	, 
	\qquad\qquad
	\delta_1 = \tfrac{a}{2} \bigl ( -1 , +\sqrt{3} \bigr )
	,
	\qquad\qquad
	\delta_2 = \tfrac{a}{2} \bigl ( -1 , -\sqrt{3} \bigr )
	,
\end{align*}
connect \emph{nearest-neighbor sites} belonging to different sublattices. The presence of two atoms per unit cell can be described by an internal degree of freedom usually referred to as isospin. Hence, if we ignore the electron's spin, the relevant Hilbert space is $\ell^2(\Gamma) \otimes \C^2$, the space of square summable sequences on the lattice $\Gamma$ with an internal $\C^2$ isospin degree of freedom. 

The simplest model which includes only nearest-neighbor hopping is defined by the hamiltonian 
\begin{align}
	T(q_1,q_2) = \left (
	\begin{matrix}
		0 & \id_{\ell^2(\Gamma)} + q_1 \, \mathfrak{s}_1 + q_2 \, \mathfrak{s}_2 \\
		\id_{\ell^2(\Gamma)} + q_1 \, \mathfrak{s}_1^* + q_2 \, \mathfrak{s}_2^* & 0 \\
	\end{matrix}
	\right )
	\label{intro:eqn:nn_hamiltonian}
\end{align}
where for $j = 1 , 2$ the operators 
\begin{align*}
	(\mathfrak{s}_j \psi)_{\gamma} := \psi_{\gamma - \gamma_j}
\end{align*}
are \emph{shifts} by $\gamma_j$ and the $q_j \in \R$ are amplitudes which quantify the hopping to nearest neighbors located in adjacent unit cells. We have fixed the hopping amplitude corresponding to shifts by $\delta_0$ to $1$ by fixing a suitable energy scale. The isotropic case $q_1 = q_2 = 1$ is the standard tight-binding model for graphene (up to a rescaling in energy of order $\approx \unit[-2.8]{eV}$). In configurations $(q_1,q_2)$ close to $(1,1)$, the band spectrum of \eqref{intro:eqn:nn_hamiltonian} has two conical intersections and no spectral gaps. 

In this framework, we assume the net effect of applied strains is captured as a change of hopping parameters $(q_1,q_2)$. The range of validity of this approximation has been studied extensively \cite{pereira_castro:2009,ribeiro_pereira_peres_briddon_castro:2009}, and in our units, it suffices to consider the $q_j$ in the range $[0,2]$. The dependence of the spectrum of $T(q_1,q_2)$ on the hopping parameters is well-known \cite{Hasegawa_Konno_Nakano_Kohomoto:zero_modes_honeycomb:2006} (\cf Figure~\ref{intro:figure:spectrum_H_nn}): $\sigma \bigl ( T(q_1,q_2) \bigr )$ is symmetric around the zero energy, and thus the relevant Fermi energy $\EFermi = 0$ lies directly where the spectral gap will open. 
\begin{figure}[t]
	\hspace{1mm}
	\resizebox{!}{56mm}{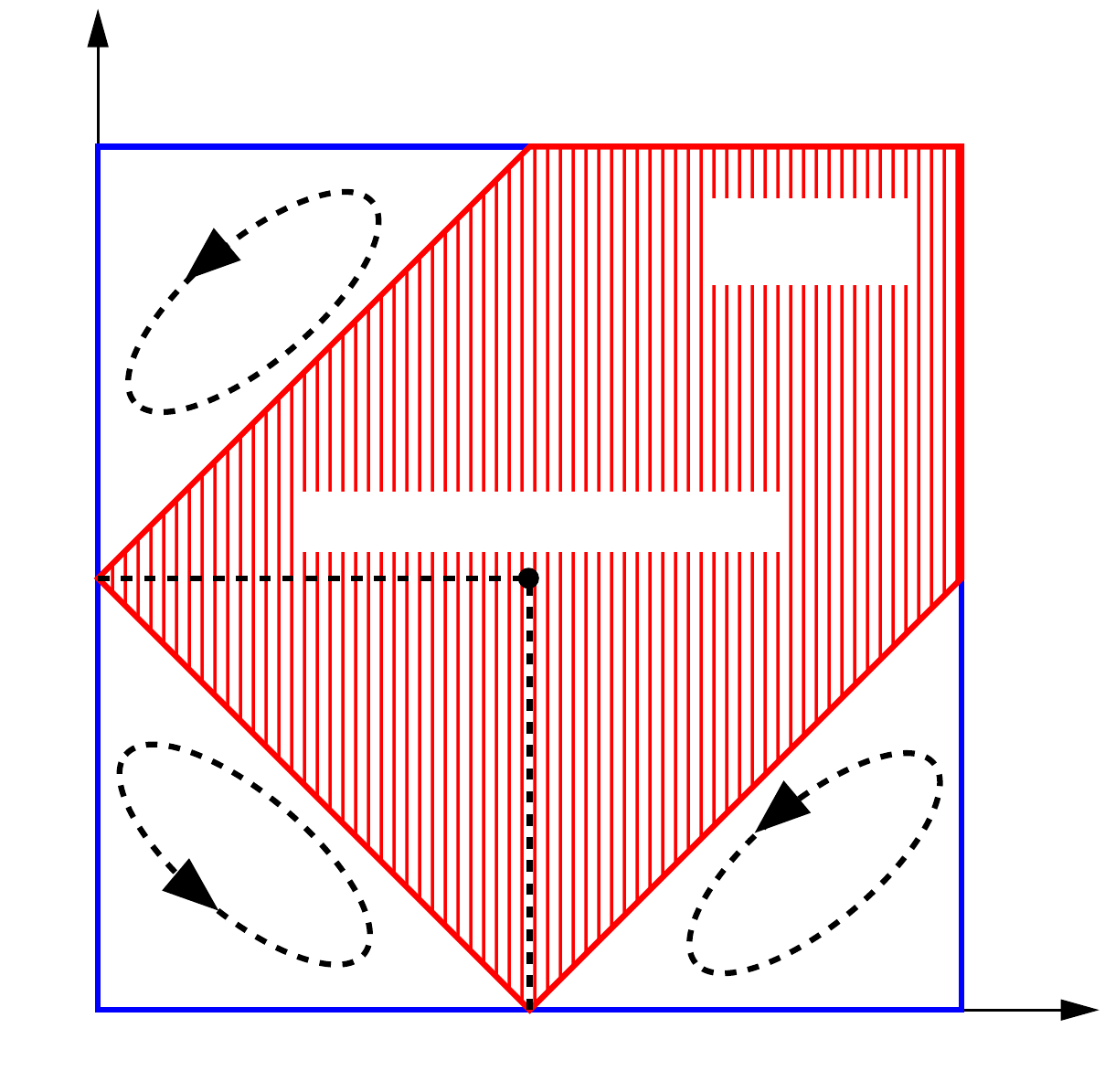}\qquad\quad \resizebox{!}{52mm}{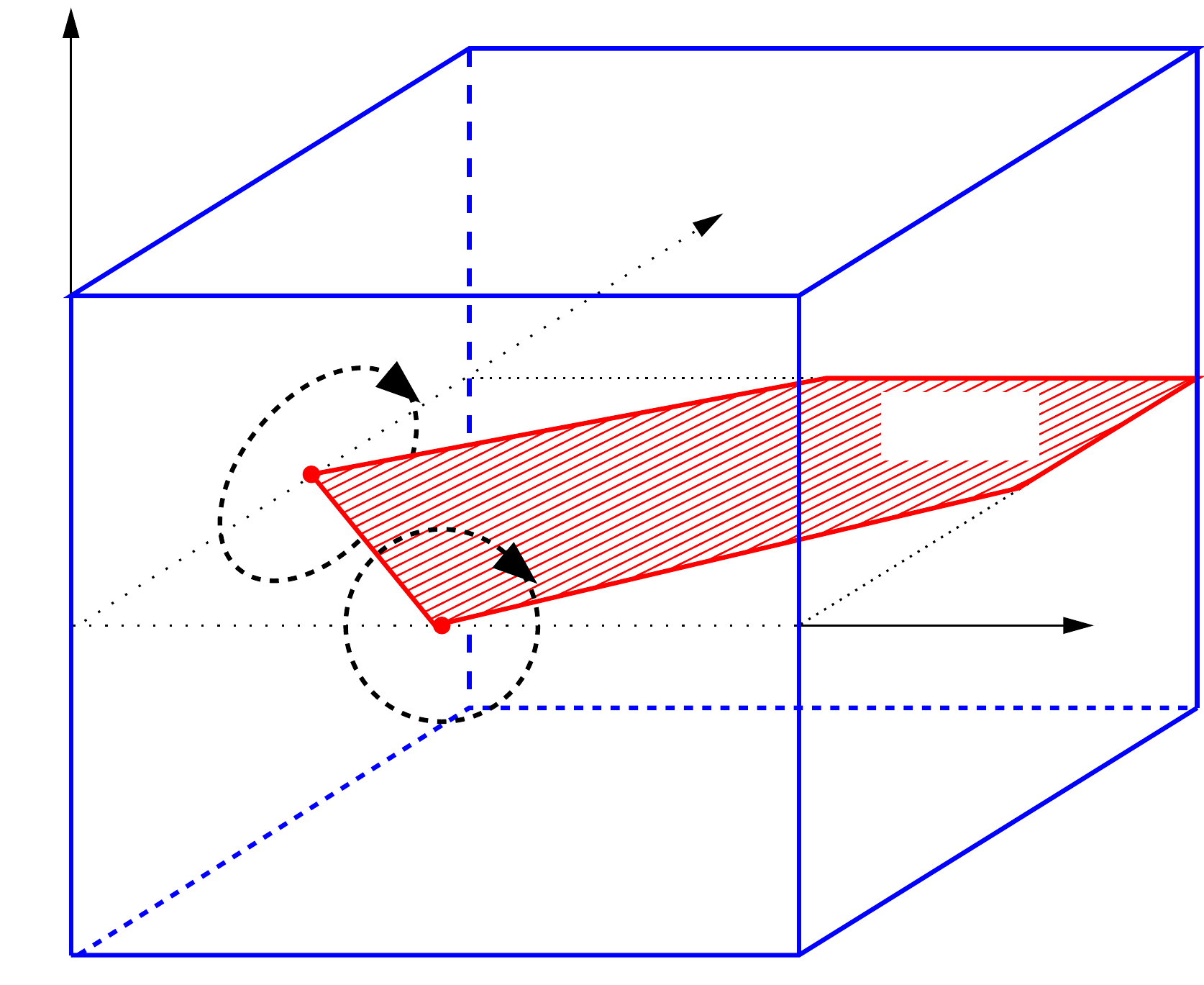}
	\hspace{0mm}
	\\[-1.3em]
	\makebox[3mm]{}(a) \makebox[60mm]{}(b)
	\caption{(a) Representation of the parameter space $[0,2]^2$ for the hamiltonian $T(q_1,q_2)$ given in \eqref{intro:eqn:nn_hamiltonian} at Fermi energy $\EFermi = 0$. For values of the parameters in the (closed) red region the system is gapless. The gapped part of the parameter space is made by the three disjoint triangular regions in white and each of this region is simply connected.
	(b) Representation of the parameter space $Q:=[0,2]^2 \times [-1,1]$ for the hamiltonian $H(q_1,q_2,q_3)$ given in \eqref{intro:eqn:uniaxial_strain_hamiltonian} at Fermi energy $\EFermi=0$. The extra dimension $q_3$ given by the stagger perturbation allows the gapped parameter space $Q_0$ (\ie $Q$ minus the red region) to be path-connected. }
	\label{intro:figure:spectrum_H_nn}
\end{figure}
If we restrict ourselves to positive hopping parameters, then the part of the parameter space $Q_0$ where the gap is open is comprised of three disjoint simply connected components. Thus, $\pi_1(Q_0) = \{ 0 \}$ for each component and according to Corollary~\ref{intro:cor:criterion_triviality}, the piezoelectric effect has to be absent. The presence or absence of topological invariants is closely related to the symmetries of a system \cite{Altland_Zirnbauer:superconductors_symmetries:1997,Schnyder_Ryu_Furusaki_Ludwig:classification_topological_insulators:2008}: 
\emph{absence of inversion symmetry} is a necessary condition for a material to be piezoelectric. However, $T(q_1,q_2)$ has an inherent inversion symmetry. Let $\wp$ be the unitary operator defined by $(\wp \psi)_{\gamma} := \psi_{-\gamma}$. Then $\wp \mathfrak{s}_j \wp = \mathfrak{s}_j^*$ holds and a simple computation yields that if we tensor $\wp$ with the Pauli matrix $\sigma_1$, we obtain an inversion symmetry of $T(q_1,q_2)$, 
\begin{align*}
	\bigl [ T(q_1,q_2) \; , \; \wp \otimes \sigma_1 \bigr ] = 0 
	. 
\end{align*}
In other words, graphene is \emph{not intrinsically piezoelectric}. To have any hopes of seeing piezoelectric effects, graphene needs to be modified in such a way as to break its inherent inversion symmetry. One potential way to achieve this is to adsorb atoms on one side of the graphene sheet (\eg hydrogen, lithium, potassium or fluorine); the piezoelectric effect of modified graphene is then expected to be comparable to that of $3d$ piezoelectric materials \cite{Ong_Reed:engineered_piezoelectricity_graphene:2012}. The simplest way to capture this breaking of inversion symmetry in the model is to add a \emph{stagger potential} to $T$, \ie to consider the uniaxial strain hamiltonian 
\begin{align}
	H(q_1,q_2,q_3) := T(q_1,q_2) + q_3 \, \left (
	\begin{matrix}
		+ \id_{\ell^2(\Gamma)} & 0 \\
		0 & - \id_{\ell^2(\Gamma)} \\
	\end{matrix}
	\right )
	\label{intro:eqn:uniaxial_strain_hamiltonian}
\end{align}
instead. We take the parameter space to be $Q = [0,2]^2 \times [-1,+1]$. Now the gapped parameter region $Q_0$ is arcwise connected and has a non-trivial fundamental group $\pi_1(Q_0) \simeq \Z^2$ (\cf Proposition~\ref{nn_model:prop:pi_1_Qgap}). Hence, any loop $\eta : \Sone \longrightarrow Q_0$ can be continuously deformed into a loop which winds $n_1$ times around $(1,0,0)$ and $n_2$ times around $(0,1,0)$ (\cf Figure~\ref{intro:figure:spectrum_H_nn}~(b)), and we get 
\begin{align}
	\Delta \mathcal{P}(\eta) = n_1 \, \Delta \mathcal{P}(\eta_1) + n_2 \, \Delta \mathcal{P}(\eta_2) 
\end{align}
where $\eta_j$ are the loops indicated in Figure~\ref{intro:figure:spectrum_H_nn}~(b). In order to prove that this model supports non-trivial piezo effects, we need to show $\Delta \mathcal{P}(\eta_j) \neq 0$. In Sections~\ref{topology} and \ref{nn_model}, we develop a technique in the spirit of \cite{Kohmoto:quantization_Hall_conductance:1985} which allows us to compute the $\Delta \mathcal{P}(\eta_j)$ (and all other Chern numbers) explicitly. 
\begin{theorem}[Piezoelectric effect in the uniaxial strain model]
	There are periodic deformations $\eta$ of \eqref{intro:eqn:uniaxial_strain_hamiltonian} such that $\Delta \mathcal{P}(\eta) \neq 0$. More precisely, let 
	\begin{align*}
		\eta_1(t) := \bigl ( 1 + \eps \, \cos t , 0 , - \eps \, \sin t \bigr )
		,
		\qquad\qquad
		\eta_2(t) := \bigl ( 0 , 1 + \eps \, \cos t , - \eps \, \sin t \bigr )
		,
	\end{align*}
	be the two generators of $\pi_1(Q_0) \simeq \Z^2$ for some $\eps \in (0,1)$ and $H_{\eta_j}(t)$ the periodic deformation of the graphene Hamiltonian \eqref{intro:eqn:uniaxial_strain_hamiltonian} along $\eta_j$. Then $\Delta \mathcal{P}(\eta_1) = (1,0)$ and~$\Delta \mathcal{P}(\eta_2) = (0,1)$, and thus $\Delta \mathcal{P}(\eta) = (n_1,n_2)$ for $[\eta] = n_1 \, [\eta_1] + n_2 \, [\eta_2]$. 
\end{theorem}
For details of the calculations, we refer the interested reader to Section~\ref{nn_model}. 

\subsection{Stability under weak perturbations} 
\label{intro:randomness}
The uniaxial strain model above discussed is based on two simplifications: independence of electrons and absence of impurities. A more realistic model should include those aspects as well. Mathematically, we can include these effects by adding a potential $V$ to the uni\-axial strain hamiltonian \eqref{intro:eqn:uniaxial_strain_hamiltonian}, 
\begin{align}
	H_{\lambda}(q) := H(q) + \lambda \, V
	\label{intro:eqn:random_hamiltonian}
	. 
\end{align}
We assume $V \in \Alg$ is bounded ($\snorm{V}_{\Alg} = 1$ for simplicity); The perturbation can describe interactions between electrons in a mean-field approximation (periodic potential) as well as the effect of impurities (Anderson-type potential). The parameter $\lambda$ describes the strength of the perturbation. 
\begin{figure}[t]
	\hfil
	\resizebox{!}{70mm}{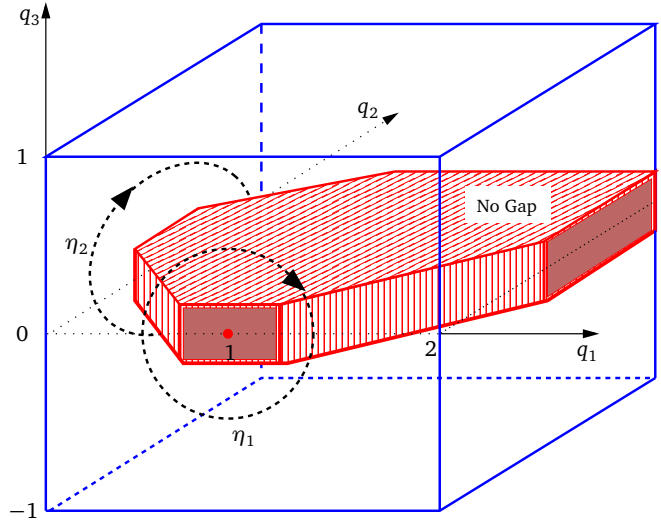}
	\hfil
	\caption{Representation of the gapped parameter space for the model \eqref{intro:eqn:uniaxial_strain_hamiltonian} perturbed by a bounded potential $\lambda V$ in the regime of a weak perturbation $\lambda\ll1$. The topology of this space agrees with the topology of the unperturbed gapped parameter space.
	}
	\label{intro:figure:random_parameter_space}
\end{figure}
Standard perturbation theory says that if the distance between $\sigma \bigl ( H(q) \bigr )$ and $\EFermi$ is greater than $g > \lambda$, then $\EFermi\notin\sigma\big(H_\lambda(q)\big)$ \cite{Kato:perturbation_theory:1995}. If $Q_{\EFermi}$ denotes the gapped parameter space for the unperturbed hamiltonian $H(q)$ then the set
\begin{align}
	Q_{\EFermi,g} := \left\{ q \in Q_{\EFermi} \; \; \big | \; \; \text{dist} \bigl ( \sigma \bigl ( H(q) \bigr ) , \EFermi \bigr ) > g \right \} 
	\label{eq:intro2_def}
\end{align}
is certainly contained in the gapped parameter space of the perturbed hamiltonian $H_\lambda(q)$. In the \emph{weak perturbation} regime $\lambda\in[0,\lambda_*]$ with $\lambda_* < g$, the space $Q_{\EFermi}$ is a \emph{deformation retract} of the space $Q_{\EFermi,g}$ and so the two have same homotopic type \cite{Hatcher:algebraic_topology:2002}. Given a loop $\eta$ in $Q_{\EFermi,g}$ one can define two periodic time-dependent and gapped operators $H_\eta(t)$ and $H_{\lambda,\eta}(t)$ according to the prescription \eqref{intro:eqn:relation_model_hamiltonian_loops}. The Fermi projections associated to these two hamiltonians are homotopic in the algebra $\Alg$ (Proposition~\ref{prop:homotopy_disorder}). As a consequence of the homotopic invariance of the King-Smith--Vanderbilt formula \eqref{intro:eqn:KSV_formula} \cite[Corollary~2]{Schulz-Baldes_Teufel:random_polarization:2012} one deduces that $H_\eta(t)$ and $H_{\lambda,\eta}(t)$ produce the same polarization vector \cite[Corollary~3]{Schulz-Baldes_Teufel:random_polarization:2012}. This fact can be stated as follows:
\begin{theorem}\label{intro:main_result:thm:perturbation_effect}
	Piezoelectric effects persist under weak perturbations. 
\end{theorem}
This result applies directly to the case of the strained graphene (see Figure~\ref{intro:figure:random_parameter_space}).

\subsection{Organization of the paper} 
First, in Section~\ref{reformulation}, we will reformulate the problem in an algebraic language. Among other things, this allows us to include effects of weak disorder just as in \cite{Schulz-Baldes_Teufel:random_polarization:2012}. 
Next, in Section~\ref{KSV} we sketch the derivation of the King-Smith--Vanderbilt formula \eqref{intro:eqn:KSV_formula} and discuss its topological nature. In particular, we connect \eqref{intro:eqn:KSV_formula} to the topology generated by the gap at $\EFermi$ in parameter space. 
Then, two-band systems are discussed in Section~\ref{topology}; we show how to exploit the fact that the Bloch bundle can be written as the pull back of a reference bundle over the $2m$-sphere (the Hopf bundle). 
Finally, we compute the polarization for the uniaxial strain model in Section~\ref{nn_model}. 

\paragraph{Acknowledgements.} 
\label{intro:acknowledgements}
G.~D{.} gratefully acknowledges support by the Alexander von Humboldt Foundation and by the grant ANR-08-BLAN-0261-01. M.~L{.} is supported by Deutscher Akademischer Austauschdienst. G.~D{.} and M.~L{.} would like to thank the Hausdorff Research Institute for Mathematics in Bonn for the invitation at the trimester program \emph{“Mathematical challenges of materials science and condensed matter physics: From quantum mechanics through statistical mechanics to nonlinear pde”} where a large part of this work was done. The authors are indebted to H.~Schulz-Baldes, S.~Teufel, A.~Giuliani and M. Porta for many interesting discussions. G.~D{.} would like to thank M.~Kohmoto and A.~Trombettoni for suggesting important references on the subject. M.~L.~thanks K.~Dayal for useful references on graphene.

\section{Algebras of observables} 
\label{reformulation}
This section serves to introduce the main features of the models in which we are interested. The use of a $C^*$-algebraic approach allows us to formulate our results both, for periodic and random models simultaneously. 

In the following we will deal only with \emph{lattice models}, \ie with systems with an underlying geometry described by $d$-dimensional lattices 
\begin{align*}
	\Gamma := \Bigl \{ \gamma\in\R^d \; \big \vert \; \gamma=\mbox{$\sum_{j = 1}^d$} n_j \, \gamma_j \;, \;\;\; \; n = (n_1 , \ldots , n_d) \in \Z^d \Bigr \} 
\end{align*}
generated by $d$ linearly independent basis vectors $\{ \gamma_1 , \ldots , \gamma_d \}$ and set $\sabs{\gamma} := \sum_{j = 1}^d \sabs{n_j(\gamma)}$. We will consistently use the notation $n(\gamma) := \bigl ( n_1(\gamma) , \ldots , n_d(\gamma) \bigr ) \in \Z^d $ for the vector of coefficients which express $\gamma \in \Gamma$ in terms of the basis $\{ \gamma_1 , \ldots , \gamma_d \}$.

In order to include internal discrete degrees of freedom like spin and isospin we will usually consider tensorized objects of the form
\begin{equation}\label{sect1:tensor}
	\mathfrak{A} := \mathfrak{B}\;\otimes\;{\rm Mat}_r(\C)
\end{equation}
where $\mathfrak{B}$ is any complex normed (or locally convex) algebra and ${\rm Mat}_r(\C)$ is the algebra of the $r\times r$ matrices with complex entries. The fact that ${\rm Mat}_r(\C)$ is a finite dimensional algebra (hence \emph{nuclear}) implies that the \emph{topological} tensor product \eqref{sect1:tensor} is uniquely defined without ambiguities \cite{Treves:topological_vector_spaces:1967}. Moreover the following identification
\begin{equation*}
	\mathfrak{B}\;\otimes\;{\rm Mat}_r(\C)\;\simeq\;{\rm Mat}_r(\mathfrak{B})
\end{equation*}
will be tacitly used when convenient.

\subsection{The algebra of periodic operators} 
\label{reformulation:shift-algebra}
There is a canonical way to have a unitary representation of the lattice $\Gamma$ in terms of \emph{shift operators} on $\ell^2(\Gamma)$: to each generating vector $\gamma_j$, $j = 1 , \ldots , d$, we define 
\begin{align*}
	\shift_j : \ell^2(\Gamma) \longrightarrow \ell^2(\Gamma) 
	, 
	\quad\quad 
	\bigl ( \shift_j \psi \bigr )(\gamma) := \psi(\gamma - \gamma_j) 
	. 
\end{align*}
The shifts commute amongst each other and we can use multi-index notation to define the group action 
\begin{align}
	\shift : \Gamma \longrightarrow \mathcal{B} \bigl ( \ell^2(\Gamma) \bigr ) 
	, 
	\quad \quad
	\gamma \mapsto \shift_{\gamma} := \prod_{j = 1}^d \shift_j^{n_j(\gamma)} 
	. 
\end{align}
Starting from the algebra of finite linear combinations of shifts 
\begin{align*}
	\mathfrak{S}_{\mathrm{fin}} := \Bigl \{ \mathfrak{a} \in \mathcal{B} \bigl ( \ell^2(\Gamma) \bigr ) \; \big \vert \; \exists N \in \N_0\;\; :\; \; \mathfrak{a} = \mbox{$\sum_{\sabs{\gamma} \leqslant N}$} a_{\gamma} \, \shift_{\gamma} \Bigr \} 
	, 
\end{align*}
we define the \emph{shift algebra} $\mathfrak{S}$ as the completion of $\mathfrak{S}_{\mathrm{fin}}$ with respect to the operator norm on $\ell^2(\Gamma)$. Including also the internal degrees of freedom one defines on the spinorial Hilbert space
\begin{align*}
	\mathcal{H} := \ell^2(\Gamma) \otimes \C^r
\end{align*}
the tensorized \emph{algebra of periodic operators}
\begin{align}
	\mathfrak{A}_{\mathrm{per}} := \mathfrak{S} \otimes {\rm Mat}_r(\C)
	\simeq {\rm Mat}_r(\mathfrak{S})
	. 
	\label{periodic_algebra}
\end{align}
To avoid confusion between elements of the `abstract' Brillouin algebra $\Alg_{\mathrm{per}}$ (to be defined in Section~\ref{reformulation:Gelfand}) and its representation $\mathfrak{A}_{\mathrm{per}}$ we denote elements of $\mathfrak{A}_{\mathrm{per}}$ with a hat. The $C^*$-algebra $\mathfrak{A}_{\mathrm{per}}$ admits a \emph{differential structure} and an \emph{integration}. The \emph{position observable} $\hat{x} := (\hat{x}_1,\ldots,\hat{x}_d)$ is the vector-valued (unbounded) operator defined component-wise on $\ell^2(\Gamma)$ by
\begin{align}
	\bigl ( \hat{x}_j \psi \bigr )(\gamma) = n_j(\gamma) \, \psi(\gamma)
	. 
 \end{align}
The position operators let us define \emph{(spatial) derivations} on $\mathfrak{A}_{\mathrm{per}}$ by
\begin{align}
	\nabla_j \hat{A} := \ii \bigl [ \hat{A} \; , \, \hat{x}_j \otimes \id_r \bigr ]
	,
	\quad\quad 
	j = 1 , \ldots , d 
	.
\end{align}
Clearly, the $\nabla_j$'s are unbounded operators as can be readily seen from 
\begin{align}
	\nabla_j( \shift_{\gamma} \otimes M ) = - \ii \, n_j(\gamma) \, \shift_{\gamma} \otimes M
	\label{reformulation:eqn:derivation_shift} 
\end{align}
for every $M \in {\rm Mat}_r(\C)$. Nevertheless, the $\nabla_j$'s are initially defined on the dense subalgebra $\mathfrak{S}_{\mathrm{fin}}\otimes {\rm Mat}_r(\C)$ and then extended to their natural domain. As can be checked by explicit computation, these derivations are symmetric, $\bigl ( \nabla_j \hat{A} \bigr )^* = \nabla_j (\hat{A}^*)$, and satisfy the Leibnitz rule $\nabla_j (\hat{A} \, \hat{B}) = \nabla_j (\hat{A}) \, \hat{B} + \hat{A} \, \nabla_j(\hat{B})$. It is also easy to check that $\nabla_j$ and $\nabla_k$ commute. The \emph{gradient} operator $\nabla:=( \nabla_1 , \ldots , \nabla_d )$ is closable on a common natural domain. Let us define the norms
\begin{align}
	\snorm{\hat{A}}_p := \sum_{\abs{\alpha} \leqslant p} \bnorm{\nabla^\alpha \hat{A}}_{\mathcal{B}(\mathcal{H})} 
	\label{reformulation:eqn:k_norm_algebra}
\end{align}
where for any $\alpha=(\alpha_1\ldots \alpha_d)\in\N_0^d$ we used a multi index notation $\nabla^{\alpha} := \nabla_1^{\alpha_1} \cdots \nabla_d^{\alpha_d}$ with $|\alpha| = |\alpha_1| + \ldots + |\alpha_d|$. The completion of the dense subalgebra $\mathfrak{S}_{\mathrm{fin}} \otimes {\rm Mat}_r(\C)$ with respect to $\norm{\cdot}_p$ leads to the Banach-$\ast$ algebra of \emph{$p$-times differentiable operators} $\Cont^p(\mathfrak{A}_{\mathrm{per}})\subset\mathfrak{A}_{\mathrm{per}}$. From the definition it follows that $\Cont^{p+1}(\mathfrak{A}_{\mathrm{per}})\subset \Cont^p(\mathfrak{A}_{\mathrm{per}})$ for all $p \in \N_0$ with the convention that $\Cont^0(\mathfrak{A}_{\mathrm{per}})\equiv\mathfrak{A}_{\mathrm{per}}$. The algebra $\Cont^1(\mathfrak{A}_{\mathrm{per}})$ is the natural domain for $\nabla$ while the {Fréchet-$\ast$ algebra} $\Cont^{\infty}(\mathfrak{A}_{\mathrm{per}}) := \bigcap_{p \in \N} \Cont^p(\mathfrak{A}_{\mathrm{per}})$ (endowed with the inductive limit topology) is an \emph{invariant} domain of $\nabla$.
\begin{remark}\label{reformulation:remark:stability_algebras_functional_calculus}
	All the algebras $\Cont^p(\mathfrak{A}_{\mathrm{per}})$ are stable under holomorphic and continuous functional calculus. Moreover, one can prove that if $\hat{H} = \hat{H}^{\ast} \in \Cont^p(\mathfrak{A}_{\mathrm{per}})$ and $f \in \Cont^{p+1}(\R)$ hold, then $f(\hat{H}) \in \mathfrak{A}_{\mathrm{per}}$, defined through continuous functional calculus, is in fact an element of $\Cont^p(\mathfrak{A}_{\mathrm{per}})$ \cite[Lemma~3.2]{Bratteli_Elliott_Jorgensen:decomposition_derivations:1984}. This does not only hold for $\Alg_{\mathrm{per}}$, but also all non-commutative algebras which we will work with in this paper. 
\end{remark}
The second relevant structure, the integration, is defined on $\mathfrak{A}_{\mathrm{per}}$ by the so-called \emph{trace-per-unit-volume} $\tracevol$. Let $\delta_{\gamma}\otimes e_j$ be the \emph{canonical basis} of $\mathcal{H}$ where $\{e_1,\ldots,e_r\}$ denotes the canonical basis of $\C^r$ and $\delta_{\gamma}:= (\delta_{\gamma,\gamma'})_{\gamma' \in \Gamma}$ is the $\ell^2(\Gamma)$ normalized sequence with only one non-zero entry at the label $\gamma$. Then the trace-per-unit-volume of $\hat{A} \in \mathfrak{A}_{\mathrm{per}}$ is given by
\begin{align}
	\tracevol(\hat{A})\; :=\; \sum_{j=1}^r
	\bscpro{\delta_0 \otimes e_j}{\hat{A} \; \delta_0 \otimes e_j}_{\mathcal{H}}.
	\label{reformulation:eqn:trace_per_unit_volume:1}
\end{align}
This map is a $\ast$-linear functional $\tracevol : \mathfrak{A}_{\mathrm{per}} \longrightarrow \C$ with the trace property \cite[Lemma~1]{Schulz-Baldes_Teufel:random_polarization:2012}. The name trace-per-unit-volume is justified by the following observation: Let $\{ \Gamma_n \}_{n \in \N}$ be any \emph{Følner sequence} \cite{Folner:1955} of bounded subsets of the lattice $\Gamma$ such that $\Gamma_n \subset \Gamma_{n+1}$ and $\Gamma_n \nearrow \Gamma$ and denote with $\sabs{\Gamma_n}$ the cardinality of the finite set $\Gamma_n$. If one introduces the orthogonal projections
\begin{align*}
	\hat{\chi}_{\gamma} := \sopro{\delta_{\gamma}}{\delta_{\gamma}} \otimes \id_r 
	,
	\qquad\quad
	\hat{\chi}_{\Gamma_n} := \bigoplus_{\gamma \in \Gamma_n} \hat{\chi}_{\gamma}
	,
\end{align*}
and if one observes that $\hat{\chi}_{\gamma} = (\shift_{\gamma} \otimes \id) \, \hat{\chi}_{0} \, (\shift_{\gamma} \otimes \id_r)^*$, one can verify from \eqref{reformulation:eqn:trace_per_unit_volume:1} that
\begin{align*}
	\tracevol(\hat{A}) = \trace_{\mathcal{H}} \bigl ( \hat{\chi}_{0} \, \hat{A} \, \hat{\chi}_{0} \bigr )
	= \frac{1}{\sabs{\Gamma_n}} \trace_{\mathcal{H}} \bigl ( \hat{\chi}_{\Gamma_n} \, \hat{A} \, \hat{\chi}_{\Gamma_n} \bigr )
	= \lim_{n \rightarrow \infty}\; \frac{1}{\sabs{\Gamma_n}} \trace_{\mathcal{H}} \bigl ( \hat{\chi}_{\Gamma_n} \, \hat{A} \, \hat{\chi}_{\Gamma_n} \bigr ) 
	.
\end{align*}
The calculation uses in a crucial way the cyclicity of the trace and the translation invariance of $\hat{A} \in \mathfrak{A}_{\mathrm{per}}$, namely $\bigl [ \shift_{\gamma} \otimes \id_r \, , \, \hat{A} \bigr ] = 0$ for all $\gamma \in \Gamma$. The last term in the above equality provides the usual definition of the trace-per-unit-volume (see \cite{Veselic:integrate_density_states:2008} for a general review) and the equality does not depend on the particular choice of F{\o}lner sequence.

\subsection{The Brillouin algebra for periodic lattice systems} 
\label{reformulation:Gelfand}
According to \eqref{periodic_algebra} the non-commutative part of the algebra $\mathfrak{A}_{\mathrm{per}}$ comes entirely from the matricial factor ${\rm Mat}_r(\C)$ since the algebra $\mathfrak{S}$ is commutative. This last observation allows us to use the Gelfand-Naimark theorem \cite{Hoermander:complex_analysis:1990}: it establishes a $C^*$-algebra isomorphism between $\mathfrak{S}$ and the $C^*$-algebra of continuous functions $\Cont\big(\text{Spec}(\mathfrak{S})\big)$ where the \emph{algebraic spectrum} $\text{Spec}(\mathfrak{S})$ is a compact topological Hausdorff space. 
Since the $C^\ast$-algebra $\mathfrak{S}$ is generated by the $\shift_j$, we can use \cite[Proposition~5.5]{DeNittis_Panati:topological_Bloch_Floquet:2012} to write 
\begin{align*}
	\mathrm{Spec}(\mathfrak{S}) 
	= \prod_{j = 1}^d \sigma(\shift_j)
	= \Sone \times \cdots \times \Sone = \T^d 
	. 
\end{align*}
The \emph{Gelfand isomorphism} $i_G : \mathfrak{S} \longrightarrow \Cont(\T^d)$ is uniquely defined by its action on the generators $\shift_j \mapsto \e^{- \ii k_j}$, and it maps $\sum_{\gamma \in \Gamma} c_{\gamma} \, \shift_{\gamma}$ onto the corresponding trigonometric polynomial. 

Reversing the direction of the isomorphism yields a representation of the $C^\ast$-algebra $\Cont(\T^d)$ onto the $C^\ast$-algebra of operators $\mathfrak{S}$ on the Hilbert space $\ell^2(\Gamma)$. This representation extends after tensoring with the matricial part ${\rm Mat}_r(\C)$. More precisely, let us define the \emph{periodic Brillouin algebra}
\begin{align}\label{abtr_period_alg}
	\Alg_{\text{per}} := \Cont(\T^d) \otimes {\rm Mat}_r(\C).
\end{align}
Using the identifications $\Alg_{\text{per}}\simeq {\rm Mat}_r\big(\Cont(\T^d)\big)$ and $\mathfrak{A}_{\mathrm{per}}\simeq {\rm Mat}_r\big(\mathfrak{S}\big)$, and the Gelfand isomorphism $\imath_G^{-1}$ for each component, we define a faithful representation of the periodic Brillouin algebra on the algebra of periodic operators, 
\begin{align}
	\pi_{\text{per}} : \Alg_{\text{per}} \longrightarrow \mathfrak{A}_{\mathrm{per}} \subset \mathcal{B}(\mathcal{H})
	. 
	\label{rep_period_alg}
\end{align}
As we shall explain in the next subsection, this point of view extends in a natural way to the case of random operators. The representation $\pi_{\text{per}}$ can be concretely realized through the \emph{Fourier transform} 
\begin{align*}
 (\Fourier \psi)(k) := \sum_{\gamma \in \Gamma} \bigl ( \e^{- \ii k \cdot x} \psi \bigr )(\gamma) 
	= \sum_{\gamma \in \Gamma} \e^{- \ii k \cdot n(\gamma)} \, \psi(\gamma) 
\end{align*}
which is a unitary map $\Fourier : \ell^2(\Gamma) \rightarrow L^2(\T^d)$ between Hilbert spaces. A simple computation shows that $\Fourier\,\shift_j\,\Fourier^\ast=\e^{- \ii k_j}$ where the right-hand side must be interpreted as a multiplication operator on $L^2(\T^d)$. This means that the Gelfand isomorphism $\imath_G^{-1}$ is unitarily implemented by $\imath_G^{-1}(f) \equiv \Fourier^\ast\, f \, \Fourier \in \mathfrak{S}$ for each continuous function $f$ in $\Cont(\T^d)$. Tensorizing the Fourier transform by the identity matrix $\id_r$ one obtains also a unitary description of the representation \eqref{rep_period_alg}. More precisely for each continuous matrix-valued function $A \in \Alg_{\text{per}}$ one verifies that
\begin{align*}
	\pi_{\text{per}}(A) \equiv (\Fourier\otimes\id_r)^\ast \, A \; (\Fourier\otimes\id_r) =: \hat{A} \in \mathfrak{A}_{\mathrm{per}}.
\end{align*}
The fact that the algebra of operators $\mathfrak{A}_{\mathrm{per}}$ is just a faithful representation of the algebra $\Alg_{\text{per}}$ allows us to to investigate spectral and dynamical aspects directly in the algebra $\Alg_{\text{per}}$. The first advantage concerns the calculation of the spectrum. The above relation also means that $A \in \mathcal{A}_{\mathrm{per}}$ and $\hat{A} = \pi_{\mathrm{per}}(A) \in \mathfrak{A}_{\mathrm{per}}$ are unitarily equivalent, and thus $\sigma(A) = \sigma(\hat{A})$. Now the spectrum of $A$ is easily accessible since it acts as a matrix-valued multiplication operator on the \emph{fibered} Hilbert space $L^2(\T^d) \otimes \C^r$. Of particular interest is the case of a selfadjoint $\hat{H} = \hat{H}^{\ast} \in \mathfrak{A}_{\text{per}}$. In this case for each $k \in \T^d$ the operator $H(k)$ is a symmetric $r\times r$ matrix with eigenvalues $E_1(k) \leqslant \ldots \leqslant E_r(k)$. The functions $k \mapsto E_j(k)$ are called energy bands. As a standard result \cite{Reed_Simon:M_cap_Phi_4:1978} we have a complete characterization of the spectrum:
\begin{lemma}\label{reformulation:lem:spectrum_elements_Alg}
	The spectrum of any selfadjoint $\hat{H} = \hat{H}^* \in \mathfrak{A}_{\mathrm{per}}$ has empty singularly continuous components and consists of closed intervals, 
	\begin{align*}
		\sigma(\hat{H}) = I_1\; \cup \ldots \cup I_r	
	\end{align*}
	where $I_j:= \bigcup_{k \in \T^d} \bigl \{ E_j(k) \bigr \} = [ E_j^{\rm min} , E_j^{\rm max}]$.	
\end{lemma}
Also the differential structure and the integration defined on the operators algebra $\mathfrak{A}_{\mathrm{per}}$ have a counterpart on the level of the algebra $\Alg_{\text{per}}$. First of all, since any $A$ in $\Alg_{\mathrm{per}}$ can be seen as a map from the manifold $\T^d$ to the normed algebra ${\rm Mat}_r(\C)$ the $d$ partial derivatives $\partial_{k_j} A$ are well defined (assuming $A$ is sufficiently regular). A formal computation on linear combinations of generators and elementary tensor products leads to the relation 
\begin{align}
	\pi_{\text{per}}(\partial_{k_j} A) = \ii \bigl [ \pi_{\text{per}}(A) \; , \, \hat{x}_j \otimes \id_r \bigr ] = 
	\nabla_j \bigl ( \pi_{\text{per}}(A) \bigr )
	.
	\label{reformulation:eqn:i_F_intertwines_der_partial}
\end{align}
With the identification of the notation $\partial_{k_j}\equiv\nabla_j$ we can rewrite the equation \eqref{reformulation:eqn:i_F_intertwines_der_partial} as 
\begin{align}
	\pi_{\text{per}}\circ \nabla_j = \nabla_j\circ \pi_{\text{per}}
	\label{commut_gradient}
\end{align}
which means that the representation $\pi_{\text{per}}$ intertwines with the gradient $\nabla$. The above relation is well defined on a maximal domain. Let us introduce the regular subalgebras of $\Alg_{\mathrm{per}}$:
\begin{align*}
	\Cont^p(\Alg_{\mathrm{per}}) = \Cont^p(\T^d) \otimes {\rm Mat}_r(\C)\;\simeq\; {\rm Mat}_r\big(\Cont^p(\T^d)\big)
	,\qquad 
	p=0,1,\ldots,\infty 
	.
\end{align*}
Then it is straightforward to check that $\pi_{\text{per}}$ defines faithful $\ast$-algebra map
$$
\pi_{\text{per}} : \Cont^p(\Alg_{\mathrm{per}}) \; \longrightarrow \; \Cont^p(\mathfrak{A}_{\mathrm{per}}).
$$
This together with equation \eqref{commut_gradient} establishes that $\Cont^1(\Alg_{\mathrm{per}})\subset \Alg_{\mathrm{per}}$ is the natural domain for the gradient $\nabla$ and $\Cont^{\infty}(\Alg_{\mathrm{per}}) \subset \Alg_{\mathrm{per}}$ is an invariant domain.

On $\Alg_{\text{per}} \simeq {\rm Mat}_r \bigl ( \Cont(\T^d) \bigr )$ the integration $\tracevol$ involves a \emph{bona fide} integral and the trace, 
\begin{align}
	\tracevol(A) := \int_{\T^d} \frac{\dd k}{(2\pi)^d} \; \text{Tr}_{\C^r}\bigl ( A(k) \bigr )
	. 
	\label{eqn:trace_Cont_Td}
\end{align}
Here $\dd k$ is normalized such that $\int_{\T^d} \dd k = (2\pi)^d$. Also in this case one can verify (first on a dense subalgebra) the intertwining relation
\begin{align}
	\tracevol = \tracevol \circ \pi_{\text{per}}
	\label{commut_gradient}
\end{align}
between \eqref{eqn:trace_Cont_Td} and the trace-per-unit-volume \eqref{reformulation:eqn:trace_per_unit_volume:1}. In particular all properties listed in \cite[Lemma~1]{Schulz-Baldes_Teufel:random_polarization:2012} hold true also for \eqref{eqn:trace_Cont_Td}.
\begin{remark}\label{remark:notation2}
We stress that the use of the symbols $\nabla_j$ and $\tracevol$ both for the algebra $\Alg_{\text{per}}$ and its realization $\mathfrak{A}_{\mathrm{per}}$ has the great advantage of allowing us to write formulas like \eqref{intro:eqn:KSV_formula} independently of the specific realization in a given algebra. Moreover, the risk of confusion caused by this abuse of notation is minimal and, when necessary, the reference to the algebra will be mentioned explicitly.
\end{remark}
%


\subsection{Covariant families of random operators} 
\label{reformulation:covariant}
A random system on the lattice $\Gamma$ is described by hamiltonians of the form
$$
\hat{H}_{\omega} := \hat{H} + \lambda \hat{V}_{\omega}
$$
where $\hat{H}$ is a periodic operator, \ie an element of $\mathfrak{A}_{\mathrm{per}}$, $\hat{V}_{\omega}$ is a bounded operator on $\mathcal{H}$ which depends on a random parameter $\omega$ and a coupling constant $\lambda>0$. A typical example is the \emph{Anderson potential}
$$
\hat{V}_{\omega} := \sum_{\gamma\in\Gamma} \omega_{\gamma} \, \sopro{\delta_{\gamma}}{\delta_{\gamma}} \otimes \id_r 
$$
with $\omega_\gamma\in[0,1]$. The collection $\omega:=(\omega_\gamma)_{\gamma\in\Gamma}$ defines a \emph{configuration} of the disorder. The configurations take values on the space $\Omega:=[0,1]^{\Z^d}$ which turns out to be compact if endowed with the \emph{Tychonoff topology}. If all the one site configurations $\omega_\gamma$ are distributed on the interval $[0,1]$ according to the same probability measure $\dd\mu$, one can endow also $\Omega$ with the product probability measure $\dd\mathbb{P}:= \times_{\gamma \in \Gamma} \, \dd\mu$. The group $\Gamma$ acts on the topological space $\Omega$ by translations via $(\tau_{\gamma} \omega)_{\gamma'} := \omega_{\gamma' - \gamma}$. The measure $\dd\mathbb{P}$ is invariant and ergodic with respect to the group action $\tau$. Moreover, using the invariance of $H$ with respect to the translations $\shift_{\gamma} \otimes \id_r$, one can verify the \emph{covariance property}
\begin{align}
	(\shift_{\gamma} \otimes \id_r) \; \hat{H}_{\omega} \; (\shift_{\gamma} \otimes \id_r)^{\ast} &= \hat{H}_{\tau_{\gamma} \omega}
	,
	\qquad\quad
	\forall\;\gamma\in\Gamma
	.
	\label{covariance_relations}
\end{align}
The main features of the Anderson model can be used in order to provide a general definition of \emph{random lattice systems}. Following the above example, the randomness can be described by a triple $(\Omega,\dd\mathbb{P},\tau)$ where $\Omega$ is a compact Hausdorff space, $\dd\mathbb{P}$ is a borelian probability measure and $\tau$ is an action of $\Gamma$ on $\Omega$ by homomorphisms. The measure $\dd\mathbb{P}$ is required to be invariant end ergodic with respect to $\tau$. Associated with this structure we can consider family of bounded operators $(\hat{A}_{\omega})_{\omega\in\Omega}\subset\mathcal{B}(\mathcal{H})$ such that: (i) the map $\omega \mapsto \hat{A}_{\omega}$ is strongly continuous and (ii) the covariance property \eqref{covariance_relations} holds true. We refer to such a $(\hat{A}_{\omega})_{\omega\in\Omega}$ as a \emph{covariant family of random operators}. We stress that in view of \eqref{covariance_relations} periodic operators can be identified with constant covariant operators, namely with a random family such that $\hat{A}_{\omega} = \hat{A}$ for (almost) all $\omega\in\Omega$.

Instead of a particular realization, one studies the covariant family of random operators. Many of their important properties are in fact deterministic, \eg spectrum and spectral components \cite{Pastur:1980,Kunz_Souillard:1980}
\begin{align*}
	\sigma(\hat{A}_{\omega}) = \Sigma,\qquad\quad 
	\mathbb{P}\text{-a.~e. } \omega \in \Omega 
	, 
\end{align*}
and the $\mathbb{P}$-a.~s{.} existence of the trace-per-unit-volume \cite{Bellissard:K_theory_Cstar_algebras:1986}
\begin{align}
	\tracevol(\hat{A}_{\omega}) := \lim_{n \rightarrow \infty}\; \frac{1}{\sabs{\Gamma_n}} \trace_{\mathcal{H}} \bigl ( \hat{\chi}_{\Gamma_n} \, \hat{A}_{\omega} \, \hat{\chi}_{\Gamma_n} \bigr ) 
	= \int_\Omega \dd \mathbb{P} \; \trace_{\mathcal{H}} \bigl ( \hat{\chi}_{0} \, \hat{A}_{\omega} \, \hat{\chi}_{0} \bigr )
	 \label{reformulation:eqn:trace_per_unit_volume_rand_cov}
	 . 
\end{align}
Both of the above properties are consequences of the \emph{Birkhoff ergodic theorem}.

Also the notion of derivative extends to random families of operators as a $\mathbb{P}$-a.~s{.} property and gives rise to $p$-times differentiable and smooth covariant families of random operators. 

\subsection{The non-commutative Brillouin algebra for random lattice systems} 
\label{reformulation:random}

The aim of this section is to construct an `abstract' $C^\ast$-algebra which encodes all the topological and geometrical characteristics of the set of covariant random operators. This algebra can be thought of as a generalization of the $C^\ast$-algebra  of periodic observables $\Alg_{\text{per}}$ described in Section \ref{reformulation:Gelfand}. This construction has been developed by Bellissard during the 1980's \cite{Bellissard:K_theory_Cstar_algebras:1986,Bellissard:Cstar_algebras_solid_state_physics:1988}.

Given a covariant family $(\hat{A}_{\omega})_{\omega\in\Omega}$, the main idea is to consider $A(\omega,\gamma): = \bscpro{\delta_0}{\hat{A}_{\omega}\; \delta_\gamma}_{\mathcal{H}}$ as $r\times r$ \emph{matrix-valued symbols} for covariant operator families and to construct a $C^\ast$-algebra out of these symbols, given by an adequate crossed product. First one endows the topological vector space $\Cont_{\rm c}\big(\Omega\times\Gamma,{\rm Mat}_r(\C)\big)$ of continuous functions with compact support on $\Omega\times\Gamma$ and values in ${\rm Mat}_r(\C)$ with a $\ast$-algebra structure:
\begin{align}
	AB(\omega,\gamma)\;&:=\;\sum_{\gamma'\in\Gamma}A(\omega,\gamma')\,B\big(\tau_{-\gamma'} \omega , \gamma-\gamma'\big)
	, 
	\\
	A^\ast(\omega,\gamma)\;&:=\;A\big(\tau_{-\gamma} \omega , -\gamma\big)^*
	. 
	\label{reformulation:random_ast_structure}
\end{align}
For any $\omega \in \Omega$, a representation of this $\ast$-algebra on $\mathcal{H}$ is given by
\begin{align}
	\bigl ( \pi_{\omega}(A) \Psi \bigr ) (\gamma) &:= \sum_{\gamma'\in\Gamma} A\big(\tau_{-\gamma} \omega ,\gamma'-\gamma\big) \, \Psi(\gamma')
	\label{reformulation:random_rep}
\end{align}
where $\Psi\in \mathcal{H}$ and $\Psi(\gamma)\in\C^r$ for all $\gamma\in\Gamma$. From \eqref{reformulation:random_rep} it follows that the different representations $\pi_{\omega}$ are related by the covariance relation
\begin{align*}
	(\shift_{\gamma} \otimes \id_r) \;\pi_{\omega}(A)\;(\shift_{\gamma} \otimes \id_r)^{\ast} = \pi_{\tau_\gamma \omega}(A)
\end{align*}
and are strongly continuous in $\omega$. With the notation $\pi_{\omega}(A) =: \hat{A}_{\omega}$ one can see that the family of representations $\pi_{\omega}$ sends elements of the $\ast$-algebra $\Cont_{\rm c}\big(\Omega\times\Gamma,{\rm Mat}_r(\C)\big)$ to covariant families of random operators.

If we complete $\Cont_{\rm c}\big(\Omega\times\Gamma,{\rm Mat}_r(\C)\big)$ with respect to the $C^\ast$-norm
\begin{align*}
	\snorm{A} := \sup_{\omega\in\Omega}  \bnorm{\pi_{\omega}(A)}_{\mathcal{B}(\mathcal{H})} 
	, 
\end{align*}
we obtain the \emph{(non-commutative) Brillouin algebra} 
\begin{align}\label{abtr_rand_alg}
	\Alg := \bigl ( \Cont(\Omega)\rtimes\Gamma\bigr ) \otimes {\rm Mat}_r(\C)
	\simeq {\rm Mat}_r\big(\Cont(\Omega)\rtimes\Gamma\big)
\end{align}
where $\Cont(\Omega) \rtimes \Gamma$ denotes the \emph{cossed-product} $C^\ast$-algebra \cite{Williams:crossed_products:2007}. Note that since $\Omega$ is compact, the algebra $\Cont(\Omega)$ is unital. Hence, the crossed product 
\begin{align}
	\Alg_0 := \bigl ( \C \rtimes \Gamma \bigr ) \otimes \mathrm{Mat}_r(\C)
	\simeq \mathrm{Mat}_r \bigl ( \C \rtimes \Gamma \bigr ) 
	\label{reformulation:eqn:Alg_0}
\end{align}
is a $C^*$-subalgebra of $\Alg$ which consists of the elements that are independent of $\omega$. 

The algebra $\Alg$ carries a differential structure and an integration. The first is given by a gradient $\nabla:=(\nabla_1,\ldots,\nabla_j)$ defined by
\begin{align*}
	(\nabla_j A)(\omega,\gamma) := \ii \, n_j(\gamma) \, A(\omega,\gamma)
	. 
\end{align*}
The domain of this gradient is the subalgebra $\Cont^1(\Alg)$ which is the completion of the dense algebra $\Cont_{\rm c}\big(\Omega\times\Gamma,{\rm Mat}_r(\C)\big)$ with respect to the Banach norm $\snorm{\cdot}_1$ given by a formula analogous to \eqref{reformulation:eqn:k_norm_algebra}. More generally, with the usual procedure, one can define also the algebras of $p$-times differentiable elements $\Cont^p(\Alg)$ and the algebra of smooth elements $\Cont^{\infty}(\Alg)$ which is an invariant domain of $\nabla$. A simple computation provides
\begin{align}
	\pi_{\omega} \bigl ( \nabla_j A \bigr ) = \ii \, \bigl [ \pi_{\omega}(A), \hat{x}_j\otimes\id_r \bigr ] 
	= \nabla_j \big(\pi_{\text{per}}(A)\big).
	\label{reformulation:eqn:i_F_intertwines_der_partial_ran}
\end{align}
which means that the representations $\pi_{\omega}$ intertwine with the gradient $\nabla$, 
\begin{align}
\pi_{\omega} \circ \nabla_j = \nabla_j \circ \pi_{\omega}
\label{commut_gradient_rand}
. 
\end{align}
An integration is defined on $\Alg$ by the tracial states
\begin{align}
	\tracevol(A) := \int_{\Omega} \dd \mathbb{P}\; \text{Tr}_{\C^r}\big(A(\omega,0)\big)
	. 
	\label{eqn:trace_probability}
\end{align}
A comparison between the definitions \eqref{reformulation:random_rep} and
\eqref{reformulation:eqn:trace_per_unit_volume_rand_cov} provides also in this case the $\mathbb{P}$-a.~s{.} intertwining relation
\begin{align}
	\tracevol = \tracevol\circ \pi_{\omega}
	. 
	\label{commut_gradient_rand}
\end{align}

\subsection{Perturbations of periodic operators} 
\label{reformulation:perturbations}
To consider perturbed periodic operators, we start by showing how to identify $\Alg_0 \subset \Alg$ made up of elements independent of $\omega$ (\cf equation~\eqref{reformulation:eqn:Alg_0}) with $\Alg_{\mathrm{per}}$: Equation~\eqref{reformulation:random_rep} says that every $\pi_{\omega}$ maps elements of ${\Alg}_{0}$ to bounded operators in $\mathcal{H}$ which commute with translations $\shift_j \otimes \id_m$, \ie to periodic operators. 
Periodic operators in turn are represented faithfully by $\pi_{\mathrm{per}}$. Following this reasoning, we see that 
\begin{align*}
	\pi_{\omega}^{-1} \circ \pi_{\text{per}} : {\Alg}_{\text{per}} \hookrightarrow {\Alg}
\end{align*}
establishes an isomorphism between ${\Alg}_{\text{per}}$ and ${\Alg}_{0}$, and that this isomorphism does not depend on $\omega$. 

Starting with an hamiltonian $H = H^* \in {\Alg}_{\text{per}} \subset \Alg$ which has a spectral gap at $\EFermi$ (\emph{gap condition}), 
\begin{align*}
	\mathrm{dist} \bigl ( \EFermi , \sigma(H) \bigr ) \geqslant g > 0 
	, 
\end{align*}
we can perturb it by a bounded \emph{periodic} potential in such a way that $H_{\lambda} = H + \lambda V_{\text{per}}$ is still an element of ${\Alg}_{\text{per}}$ or by a bounded \emph{random} covariant potential so that $H_{\lambda} = H + \lambda V_{\omega}$ is now an element of ${\Alg}$. In both cases $H_{\lambda}$ converges to the unperturbed periodic hamiltonian $H$ in $\Alg$ as $\lambda \rightarrow 0$. Consequently, standard perturbation theory in the sense of Kato \cite{Kato:perturbation_theory:1995} guarantees the persistence of the gap as long as the perturbation is not too strong: there exists $\lambda_{\ast} < \frac{g}{2}$ such that 
\begin{align*}
	\mathrm{dist} \bigl ( \EFermi , \sigma(H_{\lambda}) \bigr ) \geqslant \frac{g}{2} 
\end{align*}
holds for all $\lambda \in [0,\lambda_{\ast}]$. This means, we can define the Fermi projection $P_{\lambda} := 1_{(-\infty,\EFermi]}(H_{\lambda})$ for all $\lambda \in [0,\lambda_{\ast}]$. By standard results, $P_{\lambda}$ is also an element of $\Alg$ (\cf Remark~\ref{reformulation:remark:stability_algebras_functional_calculus}); furthermore, if the perturbation is periodic, then the Fermi projection is also in $\Alg_{\mathrm{per}}$. Then the continuity of $\lambda \mapsto P_{\lambda}$ can also be interpreted in the following way: 
\begin{proposition}\label{prop:homotopy_disorder}
	Under the conditions listed above, $[0,\lambda_{\ast}] \ni \lambda \mapsto P_{\lambda} \in \Alg$ is a homotopy. 
\end{proposition}
\begin{proof}
	The continuity of the family of bounded operators $H_{\lambda}$ in $\lambda$ and the resolvent identity imply the local continuity of the resolvents $(H_{\lambda} - z)^{-1}$. The gap allows us to write the Fermi projection as a Cauchy integral involving a contour that can be chosen independently of $\lambda$. Hence, $\lambda \mapsto P_{\lambda}$ is continuous or, in other words, it is a homotopy. 
\end{proof}
The importance of this result resides in the fact that all physical quantities which depend only on the homotopy class of the spectral projection can be computed in the limit of zero disorder. This meta result is generally known as \emph{stability under weak perturbation}.

\section{The King-Smith--Vanderbilt formula} 
\label{KSV}
The model hamiltonians $H : Q \longrightarrow \mathcal{A}$ we are interested in are parametrized by a parameter space $Q$. The latter is always a suitable path-connected $\Cont^p$-submanifold of $\R^N$ where notions such as taking derivatives are tacitly inherited from $\R^N$.\footnote{Indeed, one may also consider parameter spaces which are Riemannian manifolds, but for our intents and purposes, this is not necessary.} $\mathcal{A}$ is a $\ast$-algebra. We shall always make the following technical 
\begin{assumption}\label{assumption:hamiltonian}
	We assume $H : Q \subset \R^N \longrightarrow \Cont^1(\mathcal{A})$ is a $\Cont^p$ map with $p \geqslant 3$ taking values in the selfadjoint elements of $\Cont^1(\mathcal{A})$. 
\end{assumption}
Let us pick a Fermi energy $\EFermi$ and consider the gapped parameter set $Q_{\EFermi}$ defined as in the introduction. Then the relation~\eqref{intro:eqn:relation_model_hamiltonian_loops} mediates between loops $\eta \in \Cont(\Sone,Q_{\EFermi})$ and time-dependent, $T$-periodic hamiltonians $H_{\eta}(t)$. Note that if $Q_{\EFermi}$ has several connected components, then $\pi_1(\Qgap)$ is the direct sum of the fundamental groups of the connected components.

\subsection{Derivation from first principles} 
\label{KSV:derivation}
To give a self-contained presentation, we will sketch the derivation of \eqref{intro:eqn:KSV_formula}. The assumption $H_{\eta}(t) \in \Cont^1(\mathcal{A})$ implies that the current operator 
\begin{align*}
	\jmath_{\eta}(t) := \nabla H_{\eta}(t)
\end{align*}
is \emph{bounded}. 

The dynamical polarization is the expectation value of the charge transported over one period, and a quick computation \cite[Proposition~4]{Schulz-Baldes_Teufel:random_polarization:2012} yields 
\begin{align*}
	\Delta \mathcal{P}_{\mathrm{dyn}}(\eta) :& \negmedspace= \int_0^T \dd t \, \tracevol \left ( P(t) \; \nabla H_{\eta}(t) \right ) 
	\\
	&= \ii \, \int_0^T \dd t \; \tracevol \Bigl ( P(t) \; \bigl [ \partial_t P(t) \; , \, \nabla_j P(t) \bigr ] \Bigr ) 
	, 
\end{align*}
where $P(t)$ is the solution to the Liouville equation with initial state $P_{\eta}(0)$ as given by equation~\eqref{intro:eqn:Fermi_projection}. 

Assuming the deformation is sufficiently slow and regular, we can approximate $\Delta \mathcal{P}_{\mathrm{dyn}}(\eta)$ with $\Delta \Pol(\eta)$ by replacing the time-evolved projection $P(t)$ with the Fermi projection $P_{\eta}(t)$: if $t \mapsto H_{\eta}(t)$ is $\Cont^p$ as a $T$-periodic map from $\R$ to $\Cont^1(\Alg)$, then \cite[Theorem~1]{Schulz-Baldes_Teufel:random_polarization:2012} states that the error is of $(p-2)$th order in the adiabatic parameter $\eps$, 
\begin{align*}
	\Delta \Pol_{\mathrm{dyn}}(\eta) = \Delta \Pol(\eta) + \order(\eps^{p-2}) 
	. 
\end{align*}
%

\subsection{The polarization as a topological quantity} 
\label{KSV:topological}
A second, and for our purposes equally important result, \cite[Corollary~2]{Schulz-Baldes_Teufel:random_polarization:2012}, says that $\Delta \Pol$ is invariant under $\Cont^1$-homotopies of \emph{projections} which gives us leeway in how to calculate $\Delta \Pol$. The remainder of Section~\ref{KSV} serves to show that instead of looking at homotopies of \emph{projections}, it suffices to look at homotopies in \emph{parameter space}. 
\begin{lemma}\label{KSV:lem:homotopic_loops_homotopic_projections}
	If $\eta$ and $\eta'$ are in the same equivalence class of the $p$-regular homotopy group $\pi_1^p(\Qgap)$, then also the Fermi projections $P_{\eta}$ and $P_{\eta'}$ are $\Cont^p$-homotopic. 
\end{lemma}
\begin{proof}
	Let $\eta , \eta' \in \Cont^p(\Sone,\Qgap)$ be two loops with $[\eta]_p = [\eta']_p \in \pi_1^p(\Qgap)$. By definition, there exists a $\Cont^p$-homotopy 
	\begin{align*}
		\Lambda : [0,1] \longrightarrow \Cont^p(\Sone,\Qgap) 
	\end{align*}
	which connects $\eta = \Lambda(0)$ with $\eta' = \Lambda(1)$. The continuity of $(s,t) \mapsto H_{\Lambda(s)}(t)$ ensures the inner and outer continuity of $\sigma \bigl ( H_{\Lambda(s)}(t) \bigr )$ (see \eg \cite[Corollary~2.6]{Mantoiu_Purice:continuity_spectra:2009}); moreover, the resolvent $(s,t) \mapsto \bigl ( H_{\Lambda(s)}(t) - z \bigr )^{-1}$ inherits the $\Cont^p$ regularity of $q \mapsto H(q)$ and $\Lambda(s)$. Hence, writing $P_{\Lambda(s)}(t)$ as a Cauchy integral, we see that the map $(s,t) \mapsto P_{\Lambda(s)}(t)$ is also $\Cont^p$. 
	
	Now we cover $[0,1]$ with finitely many open intervals $\{ \mathcal{V}_{\alpha} \}$ such that 
	\begin{align}
		\sup_{t \in \Sone} \bnorm{P_{\Lambda(s)}(t) - P_{\Lambda(s')}(t)} < 1 
		\label{topology_Bloch:eqn:norm_distance_projections}
	\end{align}
	holds for all $s , s' \in \mathcal{V}_{\alpha}$. 
	
	Let us initially assume \eqref{topology_Bloch:eqn:norm_distance_projections} holds for all $s , s' \in [0,1]$. Then this condition implies the existence of a family of unitaries \cite[equation~(4.38)]{Kato:perturbation_theory:1995} 
	\begin{align}
		U(s;t) := \Bigl ( P_{\Lambda(s)}(t) \, P_{\eta}(t) + \bigl ( \id - P_{\Lambda(s)}(t) \bigr ) \, \bigl ( \id - P_{\eta}(t) \bigr ) \Bigr ) \, \Bigl ( \id - \bigl ( P_{\Lambda(s)}(t) - P_{\eta}(t) \bigr )^2 \Bigr )^{- \nicefrac{1}{2}} 
		\label{KSV:eqn:unitary_intertwiner_explicity}
	\end{align}
	which intertwines $P_{\eta}(t)$ and 
	\begin{align}
		P_{\Lambda(s)}(t) = U(s;t) \, P_{\eta}(t) \, U(s;t)^* 
		. 
		\label{KSV:eqn:projection_homotopy}
	\end{align}
	This unitary $U(s;t)$ inherits the $\Cont^p$ regularity from $P_{\Lambda(s)}(t)$ and $P_{\eta}(t)$. 
	
	Let us return to the general case: If we need several open intervals $\{ \mathcal{V}_{\alpha} \}$ to cover $[0,1]$ so that \eqref{topology_Bloch:eqn:norm_distance_projections} is satisfied on each of them, the above procedure yields a collection of homotopies $\{ P_{\alpha}(s;t) \}$ where each of the $P_{\alpha}(\cdot;t)$ is only defined on $\mathcal{V}_{\alpha}$. Gluing these homotopies together yields a homotopy $P(s;t)$ connecting $P_{\eta}(t)$ and $P_{\eta'}(t)$ which is $\Cont^p$ almost everywhere, but only continuous at the gluing points. However, we can invoke \cite[Corollary~17.8.1]{Bott_Tu:diff_forms_algebraic_topology:1982} and modify the homotopy $P(\cdot;t)$ to make it $\Cont^p$ everywhere. 
\end{proof}
According to \cite[Section~3]{Schulz-Baldes_Teufel:random_polarization:2012}, the quantity $\Delta \Pol(\eta)$ given by \eqref{intro:eqn:KSV_formula} (for a fixed $\eta$) is a \emph{two-cocycle} on the {extended} $C^\ast$-algebra $\widehat{\Alg}:=\Cont(\Sone)\otimes\Alg$ endowed with the {extended} gradient $\widehat{\nabla}:=(\ii \partial_t,\nabla)$ and the extended trace $\widehat{\tracevol}:=\int_0^T \dd t \, \tracevol(\cdot)$. Such an object provides the non-commutative analog of the Chern invariant in the spirit of \emph{non-commutative differential calculus}. More precisely, $\widehat{\tracevol}$ can be seen a map between the $K$-group $K_0(\widehat{\Alg})$ and $\Z$ \cite{Connes:noncommutative_geometry:1994,Varilly_Figueroa_Gracia_Bondia:noncommutative_geometry:2001}: $\widehat{\tracevol}$ applied to a projection in $\widehat{\Alg}$ yields an integer, and if two projections are $K_0$-equivalent, they are mapped to the same integer. On the other hand, Lemma~\ref{KSV:lem:homotopic_loops_homotopic_projections} says that for $\eta$ and $\eta'$ from the same equivalence class of the $p$-regular homotopy group $\pi_1^p(\Qgap)$, the Fermi projections $P_{\eta}$ and $P_{\eta'}$ are homotopic in $\Alg$, and so in the same $K_0$-class. This leads to the following result:
\begin{proposition}\label{KSV:prop:Delta_Pol_1_morphism}
	For any integer $p \geqslant 1$ equation~\eqref{intro:eqn:KSV_formula} induces a group morphism 
	\begin{align*}
		\Delta \Pol_{\ast}^p : \Cont^p(\Sone,\Qgap) \longrightarrow \Z^d
		, 
		\quad
		[\eta]_p \mapsto \Delta \Pol_{\ast}^p \bigl ( [\eta]_p \bigr ) := \Delta \Pol(\eta)
		. 
	\end{align*}
\end{proposition}
The linearity of $\Delta \Pol_{\ast}^p$ follows directly from the linearity of $\Delta \Pol$; also $\Delta \Pol_{\ast}^p \bigl ( [0]_p \bigr ) = 0$ is immediate, because $\partial_t P_0 = 0$ for Fermi projection associated to the constant loop $0$. 

\subsection{The periodic case: the Bloch bundle} 
\label{KSV:Bloch_bundle}
In the periodic case, the polarization $\Delta \Pol$ can be seen as Chern numbers associated to the so-called \emph{Bloch bundle}, and the $\Cont^1$-diffeotopy invariance is but one consequence of this fact. 
Using the notation of Section~\ref{reformulation:Gelfand}, the Fermi projection $P_{\eta}(t) \in \Alg_{\mathrm{per}} \simeq \mathrm{Mat}_r \bigl ( \Cont(\T^d) \bigr )$ can be viewed as a family of projections $P_{\eta}(k,t)$ on $\C^r$ indexed by $(k,t) \in \T^d$. Then in analogy to \cite{Nenciu:exponential_loc_Wannier:1983,Panati:triviality_Bloch_bundle:2006,DeNittis_Lein:exponentially_loc_Wannier:2011}, the disjoint union 
\begin{align}
	\Blochb(\eta) := \bigsqcup_{(k,t) \in \T^{d+1}} \ran P_{\eta}(k,t) 
	\label{topology_Bloch:eqn:Bloch_bundle}
\end{align}
defines the so-called \emph{Bloch bundle}. The continuity of $(k,t) \mapsto P_{\eta}(k,t)$ directly implies 
\begin{align*}
	\pi : \Blochb(\eta) \longrightarrow \T^{d+1}
\end{align*}
is a vector bundle. This construction is nothing but a particular application of the \emph{Serre-Swan theorem} \cite{Swan:vector_bundles_proj_modules:1962}. Since homotopic loops define $K_0$-equivalent projections in the module $\Cont(\Sone)\otimes \mathrm{Mat}_r\bigl ( \Cont(\T^d) \bigr ) \simeq \mathrm{Mat}_r \bigl ( \Cont(\T^{d+1}) \bigr )$ (\cf Lemma~\ref{KSV:lem:homotopic_loops_homotopic_projections}) and isomorphic vector bundles arise from $K_0$-equivalent projections \cite{Husemoller:fiber_bundles:1966}, an immediate consequence is 
\begin{proposition}\label{topology_Bloch:prop:Bloch_bundle_homotopy_class}
	The Bloch bundle $\pi : \Blochb(\eta) \longrightarrow \T^{d+1}$ depends only on the homotopy class $[\eta] \in \pi_1(\Qgap)$, \ie if $\eta \sim \eta'$, then the bundles $\Blochb(\eta)$ and $\Blochb(\eta')$ are isomorphic. In particular, all Chern classes agree, $c_n \bigl ( \Blochb(\eta) \bigr ) = c_n \bigl ( \Blochb(\eta') \bigr )$. 
\end{proposition}
We will now connect $\Delta \Polad$ to the Chern class: the homology group $H_2(\T^{d+1}) \cong \Z^{\frac{1}{2}d(d+1)}$ \cite[Section~V.C]{DeNittis_Lein:exponentially_loc_Wannier:2011} is generated by the $2$-tori 
\begin{align*}
	\T^2_{j,n} :& \negmedspace= \left \{ (k,t) \in \T^{d+1} \; \; \big \vert \; \; \bigl ( k_1 , \ldots , \cancel{k_j} , \ldots , \cancel{k_n} , \ldots , k_d , t \bigr ) = \ast \; \mbox{ where } \; \ast \in \T^{d-1} \right \} 
	, 
	\\
	\T^2_{n,d+1} :& \negmedspace= \left \{ (k,t) \in \T^{d+1} \; \; \big \vert \; \; \bigl ( k_1 , \ldots , \cancel{k_n} , \ldots , k_d , \cancel{t} \bigr ) = \ast \; \mbox{ where } \; \ast \in \T^{d-1} \right \} 
	, 
\end{align*}
where $\ast \in \T^{d-1}$ can be chosen arbitrarily for each $0 \leqslant j < n \leqslant d$ separately. Here, $\bigl ( k_1 , \ldots , \cancel{k_n} , \ldots , k_d \bigr ) \in \T^{d-1}$ means we omit $k_n$. Fixing the point $\ast$ corresponds to a choice of embedding $\T^2_{j,n} \hookrightarrow \T^{d+1}$, and how we identify $\T^2_{j,n}$ with $[-\pi,+\pi)^2 \hookrightarrow [-\pi,+\pi)^{d+1}$. Clearly, this choice has no bearing on the resulting first Chern numbers 
\begin{align}
	C_{j,n} ([\eta]) := \int_{\T^2_{j,n}} c_1 \bigl ( \Blochb(\eta) \bigr ) \in \Z 
	, 
	\label{topology_Bloch:eqn:Chern_numbers}
\end{align}
and in calculations some choices are more convenient than others. If we set $C_{j,n} ([\eta]) := - C_{n,j} ([\eta])$ for $n < j$, define $C_{j,j} ([\eta]) := 0$ and arrange 
\begin{align*}
	C ([\eta]) := \left (
	\begin{array}{c c c : c}
		 & & & + \Delta \Pol_1 ([\eta]) \\
		 & \Omega_{\mathrm{B}} ([\eta]) & & \vdots \\
		 & & & + \Delta \Pol_d ([\eta]) \\ \hdashline
		 - \Delta \Pol_1 ([\eta]) & \cdots & - \Delta \Pol_d ([\eta]) & 0 \\ 
	\end{array}
	\right )
\end{align*}
in an antisymmetric matrix, then the adiabatic polarization $\Delta \Polad ([\eta])$ make up the non-trivial components of the last column. The remainder $\Omega_{\mathrm{B}} ([\eta])$ is comprised of the Chern numbers which are relevant for the quantum Hall effect. 

\subsection{The topology of the parameter space} 
\label{KSV:parameter_space}
Let us return to the general case including disorder. In view of Proposition~\ref{topology_Bloch:prop:Bloch_bundle_homotopy_class} and the $K_0$-equivalence of homotopic projections, the differentiability condition in Proposition~\ref{KSV:prop:Delta_Pol_1_morphism} seems superfluous. Indeed, for a wide variety of cases, making this distinction is not necessary because $\pi_1^1(\Qgap)$ is isomorphic to $\pi_1(\Qgap)$. This is the case, for instance, if $Q_{\EFermi}$ is an open subset of $\R^N$, because then, the inclusion map 
\begin{align*}
	\imath : \Cont^{\infty}(\Sone,Q_{\EFermi}) \longrightarrow \Cont(\Sone,Q_{\EFermi})
\end{align*}
induces an \emph{iso}morphism $\imath_{\ast} : \pi_1^{\infty}(Q_{\EFermi}) \longrightarrow \pi_1(Q_{\EFermi})$ between smooth and continuous first homotopy groups \cite[Corollary~17.8.1]{Bott_Tu:diff_forms_algebraic_topology:1982}. 

In general, the inclusion map $\imath : \Cont^1(\Sone,\Qgap) \longrightarrow \Cont(\Sone,\Qgap)$ induces only a \emph{homo}morphism
\begin{align*}
	\imath_{\ast} : \pi_1^1(\Qgap) \longrightarrow \pi_1(\Qgap) 
	. 
\end{align*}
Very often, it is not necessary to check whether $\pi_1^1(\Qgap) \simeq \pi_1(\Qgap)$ holds for a variety of choice of $\EFermi$, but it suffices to verify $\pi_1^1(Q) \simeq \pi_1(Q)$ instead. 
\begin{proposition}\label{KSV:prop:pi_1_simeq_pi_1}
	$\pi_1^1(Q) \simeq \pi_1(Q)$ $\Longrightarrow$ $\pi_1^1(\Qgap) \simeq \pi_1(\Qgap)$ $\forall E \in \R$
\end{proposition}
To prove this, we need a lemma which is interesting in its own right: 
\begin{lemma}\label{reformulation:lem:Qgap_open}
	$Q_{\EFermi} \subseteq Q$ is open with respect to the relative topology of $Q$. 
\end{lemma}
\begin{proof}
	If $Q_{\EFermi} = \emptyset$, there is nothing to prove. So let $q_0 \in Q_{\EFermi} \neq \emptyset$ be arbitrary. Then the closedness of $\sigma \bigl ( H(q_0) \bigr )$ as well as the existence of a gap at $\EFermi$ implies that for $\epsilon > 0$ small enough, the compact interval $\bigl [\EFermi - \nicefrac{\epsilon}{2},\EFermi + \nicefrac{\epsilon}{2} \bigr ]$ is fully contained in the gap, \ie 
	\begin{align*}
		\sigma \bigl ( H(q_0) \bigr ) \cap \bigl [\EFermi - \nicefrac{\epsilon}{2},\EFermi + \nicefrac{\epsilon}{2} \bigr ] = \emptyset
		, 
		\qquad \qquad 
		\sigma \bigl ( H(q_0) \bigr ) \cap \bigl ( -\infty,\EFermi+\nicefrac{\epsilon}{2} \bigr ) \neq \emptyset
		. 
	\end{align*}
	Fundamentally, the continuity of $Q \ni q \mapsto H(q) \in \mathcal{A}$ implies the inner and outer continuity of the spectra $\sigma \bigl ( H(q) \bigr )$ (see \eg \cite[Corollary~2.6]{Mantoiu_Purice:continuity_spectra:2009}). Then the outer continuity of the spectrum $\sigma \bigl ( H(q) \bigr )$ ensures the existence of an open neighborhood $\mathcal{U}_{q_0} \subset Q$ such that 
	\begin{align*}
		\sigma \bigl ( H(q) \bigr ) \cap \bigl [\EFermi - \nicefrac{\epsilon}{2},\EFermi + \nicefrac{\epsilon}{2} \bigr ] = \emptyset
		. 
	\end{align*}
	holds for all $q \in \mathcal{U}_{q_0}$. Moreover, the inner continuity of the spectra guarantees that spectrum does not suddenly collapse, \ie 
	\begin{align*}
		\sigma \bigl ( H(q_0) \bigr ) \cap \bigl ( -\infty,\EFermi+\nicefrac{\epsilon}{2} \bigr ) \neq \emptyset
	\end{align*}
	holds on a possibly smaller open neighborhood $\mathcal{U}_{q_0}$. 
	As that also implies $\EFermi \not \in \sigma \bigl ( H(q) \bigr )$ for all $q \in \mathcal{U}_{q_0}$, we have in fact $\mathcal{U}_{q_0} \subseteq Q_{\EFermi}$. Because $q_0 \in Q_{\EFermi}$ was arbitrary, this shows $Q_{\EFermi}$ is an open subset of $Q$. 
\end{proof}
\begin{proof}[Proposition~\ref{KSV:prop:pi_1_simeq_pi_1}]
	Pick any $E \in \R$. Without loss of generality, we may assume $\Qgap \neq \emptyset$. As $\Qgap$ is an open subset of the $\Cont^1$-manifold $Q$, $\imath_{\ast} : \pi_1^1(\Qgap) \longrightarrow \pi_1(\Qgap)$ is injective. It remains to show that $\imath_{\ast}$ is also surjective. 
	
	Let $\eta \in \Cont(\Sone,\Qgap)$ be a loop. Then we may also interpret $\eta$ as a loop in $Q$, and because $\pi_1^1(Q) \simeq \pi_1(Q)$, there exist many $\Cont^1$ modifications $\eta^1$ to $\eta$ as loops in $Q$. 
	
	To see that there exist also modifications which lie entirely in $\Qgap$, we note that $\Qgap$ is open by Lemma~\ref{reformulation:lem:Qgap_open}. Thus we may pick an \emph{open} tubular neighborhood $\mathcal{U}$ of the graph of $\eta$ which lies entirely in $\Qgap$. 
	Since $\mathcal{U}$ can also be seen as an open subset of the $\Cont^1$-manifold $Q$ and $\pi_1^1(Q) \simeq \pi_1(Q)$, we can find a $\Cont^1$ modification $\eta^1$ of $\eta$ which lies in $\mathcal{U}$. 
	Then by definition $[\eta^1] = [\eta] \in \pi_1(\Qgap)$. This concludes the proof. 
\end{proof}
Hence, the assumption on $Q_{\EFermi}$ in the following theorem is satisfied in many practical cases:%
\begin{theorem}[Homotopy-invariance of $\Delta \Pol$]\label{main_result:thm:Delta_Pol_Poincare_group}
	Assume $\imath_{\ast} : \pi_1^1(Q) \longrightarrow \pi_1(Q)$ is an isomorphism. Then for any $\EFermi \in \R$ and all continuously differentiable loops $\eta \in \Cont^1(\Sone,\Alg)$, the polarization depends only on the equivalence class $[\eta] \in \pi_1(Q_{\EFermi})$, \ie $\Delta \Pol$ induces a group homomorphism $\Delta \Pol_{\ast} : \pi_1(\Qgap) \longrightarrow \Z^d$ given by 
	\begin{align*}
		\Delta \Pol_{\ast}([\eta]) = \Delta \Pol(\eta) 
	\end{align*}
	where $\eta$ is a $\Cont^1$-representative of $[\eta] \in \pi_1(Q_{\EFermi})$. 
\end{theorem}
\begin{proof}
	First of all, Lemma~\ref{KSV:prop:pi_1_simeq_pi_1} states that also $\pi_1^1(\Qgap)$ and $\pi_1(\Qgap)$ are isomorphic, and we will denote the isomorphism also with $\imath_{\ast}$. Then the composition of the group morphisms $\imath_{\ast}$ and $\Delta \Pol_{\ast}^1 : \pi_1^1(\Qgap) \longrightarrow \Z^d$ (Proposition~\ref{KSV:prop:Delta_Pol_1_morphism}) yields yet another group morphism, 
	\begin{align*}
		\Delta \Pol_{\ast} := \Delta \Pol_{\ast}^{1} \circ \imath_{\ast}^{-1} : \pi_1(Q_{\EFermi}) \longrightarrow \Z^d 
		.
	\end{align*}
\end{proof}
%

\section{Two-band systems} 
\label{topology}
We have seen that for periodic lattice systems the adiabatic polarization $\Delta \Polad([\eta])$ can be seen as Chern numbers associated to the Bloch bundle. The goal of this section is to provide a technique to \emph{compute} these Chern numbers with particular interest in models where $r = 2$.

\subsection{Periodic two-band systems} 
\label{two-bands}
A \emph{periodic two-band system} is a continuous map ${H}:Q \longrightarrow \Alg_{\mathrm{per}}$
of the form
\begin{align}\label{eq:two_level}
	{H}(k;q) = {h}_0(k;q) \, \id_{2^m} + \sum_{j = 1}^{2m+1} {h}_j(k;q) \, \Sigma_j 
\end{align}
where the functions ${h}_j(\cdot;q)\in\Cont(\T^d)$ are real valued for all $q \in Q$. The $2^m\times 2^m$ matrices $\{\Sigma_1,\ldots,\Sigma_{2m+1}\}$ are selfadjoint and provide a (non-degenerate) irreducible representation of the \emph{complex Clifford algebra} ${\rm Cl}_{\C}(2m)$.\footnote{Note that the same set of matrices provides also a \emph{degenerate} representation for the Clifford algebra ${\rm Cl}_{\C}(2m+1)$ since $\Sigma_1 \cdots \Sigma_{2m}=(\ii)^m\Sigma_{2m+1}$
\cite{Lee:1948}.} Moreover the matrices $\Sigma_j$ can be explicitly constructed as tensor products of $m$ Pauli matrices. With the above assumptions $\hat{H}(q):=\pi_{\text{per}} \bigl ( H(q) \bigr )$ is a selfadjoint operator on $\mathcal{H}$ for all $q \in Q$.

$H(k;q)$ can be diagonalized explicitly: if 
\begin{align}\label{eq:two_level3}
	 \sabs{h}(k;q) := 
	 \left ( \sum_{j = 1}^{2m+1} h_j(k;q)^2 \right )^{\nicefrac{1}{2}} 
\end{align}
is strictly positive for all $k \in \T^d$, then for this $q \in Q$ the \emph{local gap condition} is verified, the two eigenprojections 
\begin{align}
	{P}_{\pm}(k;q) := \frac{1}{2} \left ( \id_{2^m} \pm \sum_{j = 1}^{2m+1}\; \frac{{h}_{j}(k;q)}{\abs{h}(k;q)} \, \Sigma_j \right ) 
	\label{topology_Bloch:eqn:fiber_projection01}
\end{align}
exist and they satisfy ${P}_{\pm}(k;q) {P}_{\mp}(k;q)=0$, ${P}_{+}(k;q)+{P}_{-}(k;q)= \id_{2^m}$. In view of $\big({H}(k;q)-{h}_0(k;q) \, \id_{2^m}\big)^2 = \sabs{h}(k;q)^2 \, \id_{2^m}$, the matrix $H(k;q)$ can have at most two distinct eigenvalues. In fact, the orthogonal projections split 
\begin{align}\label{eq:two_level_splitting}
	{H}(k;q) = H_+(k;q) \oplus H_-(k;q)
\end{align}
into its two spectral components 
\begin{align*}
	{H}_\pm(k;q) :=P_{\pm}(k;q) \; {H}(k;q) \; P_{\pm}(k;q) 
	= \bigl ( h_0 \pm \sabs{h} \bigr )(k;q) \, {P}_{\pm}(k;q)
	.
\end{align*}
For any $(k;q)$ the rank of $P_{\pm}(k;q)$ is exactly $2^{m-1}$. This can be checked by simply taking the trace of \eqref{topology_Bloch:eqn:fiber_projection01} and exploiting that the generators of the Clifford algebra $\Sigma_j$ are traceless. 
\begin{remark}
	Clearly, for $m = 1$ all models are two-band systems. 
	However, many other models such as the \emph{Kane-Mele hamiltonian} \cite{Kane_Mele:Z2_ordering_spin_quantum_Hall_effect:2005} are not of the form \eqref{eq:two_level}, because they also include products of sigma matrices $\Sigma^{(p)}_{j_1,\ldots, j_p}:=\Sigma_{j_1}\ldots\Sigma_{j_p}$ (with $1\leqslant j_1<\ldots< j_p\leqslant 2m$ and $p=2,\ldots, 2m-1$). However, if they can be written as 
	\begin{align*}
		H = H_0 + \sum_{p=2}^{2m-1} \lambda^{p-1} H_p
		, \qquad\qquad 
		H_p := \sum_{1\leqslant j_1<\ldots< j_p\leqslant 2m} h_{j_1,\ldots, j_p} \, \Sigma^{(p)}_{j_1,\ldots, j_p}
		, 
	\end{align*}
	where $H_0$ is a two-band hamiltonian, then $H$ can be viewed as a perturbation of \eqref{eq:two_level}. Thus, topological quantities associated to $H$ coincide with those of $H_0$ provided $\lambda$ is small enough (\cf Section~\ref{reformulation:perturbations}). 
\end{remark}

\subsection{The Bloch bundle of a time-dependent two-band system} 
\label{topology_Bloch:Bloch_bundle}
For any loop $\eta \in \Cont(\Sone,Q_{\EFermi})$ we can consider the \emph{time-dependent} two-band system associated to \eqref{eq:two_level} via the prescription \eqref{intro:eqn:relation_model_hamiltonian_loops}, \ie
\begin{align*}
	{H}_{\eta}(k,t) 
	= {h}_{\eta , 0}(k,t) \, \id_{2^m} + \sum_{j = 1 }^{2m+1} {h}_{\eta,j}(k,t) \, \Sigma_j 
\end{align*}
with ${h}_{\eta , j}(k,t):={h}_{j}\big(k; \eta(\nicefrac{2\pi t}{T}) \big)$. If for all $q\in\eta$ the \emph{local gap condition} $\sabs{h}(\cdot\;;q)>0$ is verified (\eg if there exists an energy $E$ such that $\eta(\Sone) \subseteq Q_E$), we can define from \eqref{topology_Bloch:eqn:fiber_projection01} the Fermi projection 
\begin{align}
	{P}_{\eta}(k,t) := P_-\big(k; \eta(\nicefrac{2\pi t}{T}) \big) 
	= \frac{1}{2}\left ( \id_{2^m} - \sum_{j = 1}^{2m+1} \frac{{h}_{\eta,j}(k,t)}{\sabs{h_{\eta}}(k,t)} \, \Sigma_j 
	\right ) 
	. 
	\label{topology_Bloch:eqn:fiber_projection}
\end{align}
The above formula describes a continuous family of orthogonal projections ${P}_{\eta}(k,t)$ on $\C^{2^m}$ indexed by $(k,t) \in \T^{d+1}$; this family defines the \emph{Bloch bundle} (\cf Section~\ref{KSV:Bloch_bundle})
\begin{align}
	\Blochb(\eta) := \bigsqcup_{(k,t) \in \T^{d+1}} \ran {P}_{\eta}(k,t) 
	\label{topology_Bloch:eqn:Bloch_bundle}
\end{align}
of rank $2^{m-1}$. An efficient way for studying the vector bundle $\Blochb(\eta)$ is to introduce a spherical para\-metri\-zation. In fact, an inspection of equation~\eqref{topology_Bloch:eqn:fiber_projection} implies the projection depends only on the \emph{unit vector} ${\rm u}(k,t)\in\mathbb{S}^{2m}$ with components ${\rm u}_j(k,t):=\nicefrac{{h}_{\eta,j}(k,t)}{\sabs{{h}}_{\eta}(k,t)}$ which suggests to write ${P}_{\eta}(k,t)$ in terms of \emph{spherical coordinates}
\begin{equation}\label{sperical_coordinates}
\begin{aligned}
	Y_1(\theta_1,\ldots,\theta_{2m-1},\varphi)&:=\sin\theta_1\;\ldots\;\sin\theta_{2m-2}\;\sin\theta_{2m-1}\;\cos\varphi&\\
	Y_2(\theta_1,\ldots,\theta_{2m-1},\varphi)&:=\sin\theta_1\;\ldots\;\sin\theta_{2m-2}\;\sin\theta_{2m-1}\;\sin\varphi&\\
	&\;\;\;\;\vdots\\
	Y_j(\theta_1,\ldots,\theta_{2m-1},\varphi)&:=\sin\theta_1\;\ldots\;\sin\theta_{2m+1-j}\;\cos\theta_{2m+2-j}&\\
	&\;\;\;\;\vdots\\
	Y_{2m}(\theta_1,\ldots,\theta_{2m-1},\varphi)&:=\sin\theta_1\;\cos\theta_2\\
	Y_{2m+1}(\theta_1,\ldots,\theta_{2m-1},\varphi)&:=\cos\theta_1&\\
\end{aligned}
\end{equation}
where the angular coordinates $\theta_1,\ldots,\theta_{2m-1}$ range over $[0,\pi]$ and $\varphi$ ranges over $[0,2\pi]$. With the spherical coordinates \eqref{sperical_coordinates} one can construct a family of orthogonal projections on $\C^{2^m}$ given by
\begin{align}
	{P}_{\mathbb{S}^{2m}}({\rm u}):= \frac{1}{2}\left ( \id_{2^m} - \sum_{j = 1}^{2m+1} Y_j({\rm u}) \, \Sigma_j 
	\right ) 
	\label{topology_Bloch:eqn:fiber_projection_spheric}
\end{align}
with ${\rm u}:=(\theta_1,\ldots,\theta_{2m-1},\varphi)$. The projections \eqref{topology_Bloch:eqn:fiber_projection_spheric} define a \emph{reference vector bundle} $\pi:\mathcal{E}_{ \mathbb{S}^{2m}} \rightarrow \mathbb{S}^{2m}$ over the sphere $\mathbb{S}^{2m}$ with total space 
\begin{align*}
	\mathcal{E}_{ \mathbb{S}^{2m}} := \bigsqcup_{{\rm u} \in \mathbb{S}^{2m}} \ran {P}_{\mathbb{S}^{2m}}({\rm u})
	\, . 
\end{align*}
This vector bundle is also known as the \emph{Hopf bundle} in the literature; It is non-trivial with Chern class $c_m(\mathcal{E}_{ \mathbb{S}^{2m}})\simeq(-1)^{m-1}\frac{(m-1)!}{\Omega_{2m}}\dd v$ where $\Omega_{2m}$ is the volume of the sphere $\mathbb{S}^{2m}$ and $\dd v\in H^{2m}_{\rm dR}(\mathbb{S}^{2m})$ is the volume form (see \cite[Exemple 1.27]{Karoubi:cyclic_homology_K_theory:1987} and \cite[Corollary 4.4]{Hatcher:vector_bundles_K_theory:2009}).

With the help of the inverse functions
\begin{equation}\label{topology_Bloch:eqn:angle_variables01}
\begin{aligned}
	\theta_{\eta,j}(k,t) &:= \arctan \left(
	\frac{\sqrt{h_{\eta , 1}(k,t)^2+{h}_{\eta , 2}(k,t)^2+\ldots+{h}_{\eta , 2m+1-j}(k,t)^2}
	}{h_{\eta , 2m+2-j}(k,t)}\right) \\
	\varphi_{\eta}(k,t) &:= 2\arctan\left( \frac{{h}_{\eta , 2}(k,t)}{\sqrt{h_{\eta , 1}(k,t)^2+{h}_{\eta , 2}(k,t)^2}+{h}_{\eta , 1}(k,t)}\right)
 \end{aligned}
\end{equation}
one obtains a continuous map
\begin{align*}
	\Phi_{\eta} : \T^{d+1}& \longrightarrow \mathbb{S}^{2m} 
	, 
	\quad 
	(k,t) \longmapsto \bigl ( \theta_{\eta,1} ,\ldots, \theta_{\eta,2m-1} , \varphi_{\eta} \bigr ) 
	, 
\end{align*}
which allows us to write ${P}_{\eta} = {P}_{\mathbb{S}^{2m}} \circ \Phi_{\eta}$. Restated in a more sophisticated way, for a given loop $\eta$ the map $\Phi_{\eta}$ reconstructs the Bloch bundle $\Blochb(\eta)$ as the \emph{pullback} \cite{Husemoller:fiber_bundles:1966} of the reference vector bundle $\mathcal{E}_{ \mathbb{S}^{2m}}$, \ie $\Blochb(\eta) \simeq \Phi_{\eta}^* (\mathcal{E}_{ \mathbb{S}^{2m}})$.

\subsection{The case of two internal degrees of freedom} 
\label{topology_Bloch:m_equal_1}
Let us now specialize to the simplest, but still non-trivial case of a system with only two internal degrees of freedom, \ie $m=1$. In this case we identify the $\Sigma_1,\Sigma_2,\Sigma_3$ with the Pauli matrices $\sigma_1,\sigma_2,\sigma_3$. With this choice the reference projector \eqref{topology_Bloch:eqn:fiber_projection} reads
\begin{align}
	{P}_{\mathbb{S}^{2}}(\theta,\varphi) = \frac{1}{2} \left (
	\begin{matrix}
		1 - \cos \theta & -\e^{- \ii \varphi} \, \sin \theta \\
		-\e^{+ \ii \varphi} \, \sin \theta & 1 + \cos \theta \\
	\end{matrix}
	\right ) 
	, 
\end{align}
and the inverse transforms \eqref{topology_Bloch:eqn:angle_variables01} can be equivalently written as
\begin{align}
	\theta_{\eta}(k,t) &= \arccos\left( \frac{{h}_{\eta , 3}(k,t)}{\sabs{{h}_{\eta}}(k,t)} \right)
	, 
	\label{topology_Bloch:eqn:angle_variable02s}
	\\
	\varphi_{\eta}(k,t) &= \arctan \left(\frac{{h}_{\eta , 2}(k,t)}{{h}_{\eta , 1}(k,t)}\right)
	\notag 
	. 
\end{align}
It is possible to provide an explicit local trivialization for the complex line bundle $\mathcal{E}_{ \mathbb{S}^{2}}$: One needs at least two charts to cover $\Stwo$, and we shall always choose what we call an \emph{$\Nvar\Svar$ covering}: let $\mathcal{C}_{\Nvar}$ and $\mathcal{C}_{\Svar}$ be two closed, contractible and mutually disjoint sets with non-empty interior which contain only either the \emph{north pole} ($\theta=0$) or the \emph{south pole} ($\theta=\pi$) pole (denoted by $\Nvar$ and $\Svar$, respectively). Then we define $\mathcal{U}_{\Nvar} := \Stwo \setminus \mathcal{C}_{\Svar}$ and $\mathcal{U}_{\Svar} := \Stwo \setminus \mathcal{C}_{\Nvar}$ which serve as our open covering of the sphere. From the definition it follows that $\mathcal{U}_{\Nvar} \cap \mathcal{U}_{\Svar} \neq \emptyset$. On the northern hemisphere, we may choose
\begin{align}
	\Psi_{\Nvar}(\theta,\varphi) := \left (
	\begin{matrix}
		\e^{- \ii \varphi} \, \sin \tfrac{\theta}{2} \\
		-\cos \tfrac{\theta}{2} \\
	\end{matrix}
	\right )
\end{align}
as local section, \ie $(\theta,\varphi) \mapsto \Psi_{\Nvar}(\theta,\varphi)$ is continuous and span $\ran P_{\Stwo}(\theta,\varphi)$ for all $(\theta,\varphi) \in \mathcal{U}_{\Nvar}$. However, this parametrization cannot be extended unambiguously up to the {south pole} $\Svar$ since $\Psi_{\Nvar}(\Svar) = \bigl ( \e^{- \ii \varphi} , 0 \bigr )$ depends non-trivially on $\varphi$. But the parametrization 
\begin{align*}
	\Psi_{\Svar}(\theta,\varphi) := \e^{+ \ii \varphi} \, \Psi_{\Nvar}(\theta,\varphi)
	= \left (
	\begin{matrix}
		\sin \tfrac{\theta}{2} \\
		-\e^{+ \ii \varphi} \, \cos \tfrac{\theta}{2} \\
	\end{matrix}
	\right )
\end{align*}
\emph{is} well-defined at $\Svar$ (but not at $\Nvar$), and $(\theta,\varphi) \mapsto \Psi_{\Svar}(\theta,\varphi)$ is a local section on $\mathcal{U}_{\Svar}$. These local sections are glued together with the transition function 
\begin{align*}
	g_{\Nvar\Svar} : \mathcal{U}_{\Nvar} \cap \mathcal{U}_{\Svar} \longrightarrow U(1) 
	, 
	\quad 
	(\theta,\varphi) \mapsto \e^{+ \ii \varphi} 
	. 
\end{align*}
We recall that all complex line bundles $\mathcal{E} \rightarrow X$ are completely classified by the first Chern class $c_1(\mathcal{E}) \in H^2(X,\Z)$. 
Moreover, the Chern–Weil theory allows us to consider $c_1(\mathcal{E})$
as a differential 2-form (modulo exact forms), \ie $c_1(\mathcal{E}) \in H_{\rm dR}^2(X)$ \cite{Milnor_Stasheff:characteristic_classes:1974}.
In the particular case of the reference line bundle $\pi :\sphereb \longrightarrow \Stwo$ we can provide an explicit representative
for $c_1(\sphereb)$ simply following \cite[pp. 71-75]{Bott_Tu:diff_forms_algebraic_topology:1982} which uses the identification between \emph{top} Chern classes of complex vector bundles and \emph{Euler classes} of the corresponding \emph{realifications}. This gives
$$
c_1(\sphereb) =\dd \zeta
$$
where $\zeta:=\{\zeta_{\Nvar}, \zeta_{\Svar}\}$ is the collection of two local $1$-forms defined by
\begin{align*}
	\zeta_{\Nvar} :& \negmedspace= 
	\begin{cases}
		- \frac{1}{\ii 2 \pi} \, \chi_{\Svar} \, \dd \ln g_{\Nvar\Svar}
		&\qquad \mbox{on\quad $\mathcal{U}_{\Nvar} \cap \mathcal{U}_{\Svar}$} \\
		0 &\qquad \mbox{on\quad $\mathcal{C}_{\Nvar}$} \\
	\end{cases}
	\\
	\zeta_{\Svar} :& \negmedspace= 
	\begin{cases}
		+ \frac{1}{\ii 2 \pi} \, \bigl ( 1 - \chi_{\Svar} \bigr ) \, \dd \ln g_{\Nvar\Svar} &\qquad \mbox{on\quad $\mathcal{U}_{\Nvar} \cap \mathcal{U}_{\Svar}$} \\
		0 &\qquad \mbox{on\quad $\mathcal{C}_{\Svar}$} \\
	\end{cases}
\end{align*}
where $\chi_{\Svar}$ is any smooth function supported in $\mathcal{U}_{\Svar}$.
 The \emph{first Chern number} of the line bundle 
$\sphereb$ is by definition
$$
C(\sphereb):=\int_{\Stwo}c_1(\sphereb)=\int_{\Stwo} \dd \zeta.
$$
Now, let $\mathcal{K}_{\Nvar}\supset \mathcal{C}_{\Nvar}$ an open set such that $\partial\mathcal{K}_{\Nvar}$ is a sufficiently regular closed path around $\Nvar$. Moreover, we assume without loss of generality that $\chi_{\Svar} \vert_{\partial \mathcal{K}_{\Nvar}} = 1$. Then we can 
use the Stoke's theorem to write: 
\begin{align}
	C(\sphereb) &
	= \int_{\mathcal{K}_{\Nvar}}\dd \zeta 
	+ \int_{\Stwo \setminus \mathcal{K}_{\Nvar}} \dd \zeta=
	 \ointclockwise_{\partial \mathcal{K}_{\Nvar}} \zeta_{\Nvar} 
	+ \ointctrclockwise_{\partial \mathcal{K}_{\Nvar}} \zeta_{\Svar} 
	\notag 
	= \ointclockwise_{\partial \mathcal{K}_{\Nvar}} (\zeta_{\Nvar}-\zeta_{\Svar}).
\end{align}
Using the explicit form of the local forms $\zeta_{\Nvar}$ and $\zeta_{\Svar}$ and of the transition function $g_{\Nvar\Svar}$ we arrive at the important formula
\begin{align}
	C(\sphereb) &
	= \frac{1}{2\pi}\ointclockwise_{\Lambda}\dd \varphi 
	\label{eq:Chern_sbundle}
\end{align}
where $\Lambda$ is any closed regular path which goes around the north pole $\Nvar$ in the counter-clockwise direction (\ie the direction of increasing $\varphi$). Since $\varphi \in [0,2\pi]$ wraps around once, we obtain $C(\sphereb)=1$, \ie $\sphereb$ is {non-trivial}.

Now we can return to the study of the Bloch bundle $\Blochb(\eta) \simeq \Phi_{\eta}^* (\mathcal{E}_{ \mathbb{S}^{2}})$. The topology of this line bundle is fully described by the first Chern class $c_1 \bigl ( \Blochb(\eta) \bigr )$ which can be reconstructed from $c_1(\sphereb)$ with the map $\Phi_{\eta}^*$. In fact, the \emph{functoriality} (or \emph{naturality}) of the characteristic classes \cite{Milnor_Stasheff:characteristic_classes:1974} implies 
\begin{align}\label{eq:first_chern_class}
	c_1 \bigl ( \Blochb(\eta) \bigr ) = c_1 \bigl ( \Phi_{\eta}^* (\sphereb) \bigr ) 
	= \Phi_{\eta}^* \, c_1(\sphereb) 
\end{align}
where the last $\Phi_{\eta}^*$ denotes the induced group morphism between the cohomology groups $H^2(\Stwo,\Z)$ and $H^2(\T^{d+1},\Z)$ (resp{.} between the de Rham groups $H_{\rm dR}^2(\Stwo)$ and $H_{\rm dR}^2(\T^{d+1})$).

The pullback structure of the Bloch line bundle $\Blochb(\eta)$ allows us to define a local trivialization inherited from $\sphereb$. The relevant covering of the base space $\T^{d+1}$ is given by
\begin{align*}
	\mathcal{U}_{\Nvar,\eta} := \Phi_{\eta}^{-1} \bigl ( \mathcal{U}_{\Nvar} \bigr ) 
	, 
	\qquad\quad
	\mathcal{U}_{\Svar,\eta} : = \Phi_{\eta}^{-1} \bigl ( \mathcal{U}_{\Svar} \bigr ) 
	.
\end{align*}
Similarly, we will define the \emph{north} and \emph{south pole variety} as 
\begin{align*}
	\Nvar_{\eta} :& \negmedspace= \Phi_{\eta}^{-1} (\{\Nvar\})
	= \bigl \{ (k,t) \in \T^{d+1} \; \vert \; \theta_{\eta}(k,t) = 0 \bigr \}
	, 
	\\
	\Svar_{\eta} :& \negmedspace= \Phi_{\eta}^{-1} (\{\Svar\})
	= \bigl \{ (k,t) \in \T^{d+1} \; \vert \; \theta_{\eta}(k,t) = \pi \bigr \}
	. 
\end{align*}
On $\mathcal{U}_{\Nvar,\eta}$ and $\mathcal{U}_{\Svar,\eta}$, the compositions $\Psi_{\sharp,\eta}(k,t):=\Psi_{\sharp} \circ \Phi_{\eta}(k,t) = \Psi_{\sharp} \bigl ( \theta_{\eta}(k,t),\varphi_{\eta}(k,t) \bigr )$, with $\sharp=\Nvar,\Svar$ respectively, are local frames which are glued together through the transition function $g_{\Nvar\Svar,\eta}:=g_{\Nvar\Svar}\circ \Phi_{\eta}= \e^{+ \ii \varphi_\eta}$. This give us an immediate criterium for the triviality of the Bloch bundle $\Blochb(\eta)$: 
\begin{theorem}
	If $\Nvar_{\eta} = \emptyset$ or $\Svar_{\eta} = \emptyset$, then the Bloch line bundle $\Blochb(\eta)$ is trivial. 
\end{theorem}
So let us proceed to compute \eqref{eq:first_chern_class}: We recall that the $\T^2_{j,n} \hookrightarrow \T^{d+1}$, with $1\leqslant j<n\leqslant d+1$, are the generators of the homology group $H_2(\T^{d+1})$ described in Section \ref{KSV:Bloch_bundle}. The integration of the 2-form $c_1 \bigl ( \Blochb(\eta) \bigr )$ over these two-dimensional surfaces produces a set of independent Chern numbers
\begin{align*}
	C_{j,n}\big([\eta]\big) &:= \int_{\T^2_{j,n}}c_1 \bigl ( \Blochb(\eta) \bigr )= \int_{\T^2_{j,n}}\Phi_{\eta}^* \, c_1(\sphereb). 
	\label{topology_Bloch:eqn:Chern_number_integral_formula01}
\end{align*}
The Chern class $c_1 \bigl ( \Blochb(\eta) \bigr )$ can be realized as the
exterior derivative of the locally defined $1$-forms $\bigl \{ \dd \zeta_{\Nvar,\eta} , \dd \zeta_{\Svar,\eta} \bigr \}$ where $\dd \zeta_{\sharp,\eta}:=\Phi_{\eta}^* (\dd \zeta_{\sharp})$ are defined on $\mathcal{U}_{\sharp,\eta}$ with $\sharp=\Nvar,\Svar$.

Now we can pick an open set $\mathcal{K}_{\Nvar,\eta} \supset \mathcal{C}_{\Nvar,\eta} := \Phi_{\eta}^{-1} (\mathcal{C}_{\Nvar})$ which contains the north pole variety $\Nvar_\eta$ with a sufficiently regular boundary $\partial \mathcal{K}_{\Nvar,\eta} \subset \mathcal{U}_{\Nvar,\eta} \cap \mathcal{U}_{\Svar,\eta}$. \emph{If} we can embed $\T^2_{j,n} \hookrightarrow \T^{d+1}$ such that $\T^2_{j,n} \cap \mathcal{K}_{\Nvar,\eta} = \emptyset$, then Stoke's theorem implies $C_{j,n}([\eta]) = 0$. Put another way, the Chern class $c_1 \bigl ( \Blochb(\eta) \bigr )$ behaves like the exact $2$-form $\dd \zeta_{\Svar,\eta}$ on $\T^2_{j,n}$. This gives yet another simple criterion: if $\Nvar_{\eta} \cap \T^2_{j,n} = \emptyset$ or $\Svar_{\eta} \cap \T^2_{j,n} = \emptyset$, then $C_{j,n}([\eta]) = 0$. 

However, if both the intersections are non-empty, then we exploit
\begin{align*}
	\partial \bigl ( \T^2_{j,n} \cap \mathcal{K}_{\Nvar,\eta} \bigr ) = \partial \bigl ( \T^2_{j,n}\setminus(\T^2_{j,n} \cap \mathcal{K}_{\Nvar,\eta}) \bigr )=\T^2_{j,n} \cap \partial \mathcal{K}_{\Nvar,\eta}
\end{align*}
and $\dd \ln g_{\Nvar\Svar} \circ \Phi_{\eta} = \ii \, \dd \varphi_{\eta}$ to obtain a description of $C_{j,n}([\eta])$ as a winding number, 
\begin{align}
	C_{j,n}([\eta]) &= \ointclockwise_{\partial (\T^2_{j,n} \cap K_{\Nvar,\eta})} \bigl ( \zeta_{\Svar,\eta} - \zeta_{\Nvar,\eta} \bigr )
	= \frac{1}{2\pi} \ointclockwise_{\Lambda_{j,n}} \dd \varphi_{\eta} 
	. 
	\label{topology_Bloch:eqn:Chern_number_integral_formula}
\end{align}
Here $\Lambda_{j,n}$ is any closed regular path which goes around the north singular set $\Nvar_\eta\cap\T^2_{j,n}$ in the counter-clockwise direction.

\section{The uniaxial strain model} 
\label{nn_model}
\renewcommand{\Qgap}{Q_0} 
We now compute \eqref{topology_Bloch:eqn:Chern_number_integral_formula} for the simplest kind of model for piezo effects in graphene, the uniaxial strain model. It allows for nearest neighbor interactions with three hopping parameters $q_0$, $q_1$ and $q_2$ as well as a stagger parameter $q_3$: 
\begin{align}
	\hat{H}(q_0,q_1,q_2,q_3) = \left (
	\begin{matrix}
		+ q_3 \id_{\ell^2(\Gamma)} & q_0 \, \id_{\ell^2(\Gamma)} + q_1 \, \mathfrak{s}_1 + q_2 \, \mathfrak{s}_2 \\
		\overline{q_0} \, \id_{\ell^2(\Gamma)} + \overline{q_1} \, \mathfrak{s}_1^* + \overline{q_2} \, \mathfrak{s}_2^* & -{q_3} \id_{\ell^2(\Gamma)} \\
	\end{matrix}
	\right )
	\label{intro:eqn:nn_hamiltonianXX}
\end{align}
In principle, the parameters can be complex, but to ensure selfadjointness of $\hat{H}$, $q_3$ needs to be real. 

To understand the action of this hamiltonian, we recall Figure~\ref{intro:figure:honeycomb}: the honeycomb lattice consists of two sublattices, and each (white/black) atom has three (black/wite) nearest neighbors. The offdiagonal terms describe nearest-neighbor hopping. If an electron is initially located at the white atom at $\gamma$, it hops to its neighboring black sites located at $\gamma$, $\gamma - \gamma_1$ and $\gamma - \gamma_2$ with rates $q_0$, $q_1$ and $q_2$ (green vectors). 
Applying stagger $q_3 \neq 0$ implies that white and black atoms within a unit cell are no longer equivalent, thus breaking the inversion symmetry.

\subsection{Description of the gapped parameter space} 
\label{nn_model:gap_structure}
Instead of the periodic two-band operator \eqref{intro:eqn:nn_hamiltonianXX}, we will study its counterpart in the Brillouin algebra $\Alg_{\text{per}}$, namely 
\begin{align}\label{nn_model:eqn:parametric_hamiltonian}
	H(k;q) &=
	\Re \bigl ( \varpi(k;q) \bigr ) \, \sigma_1 + \Im \bigl ( \varpi(k;q) \bigr ) \, \sigma_2 + q_3 \, \sigma_3
	, 	
\end{align}
where we have introduced the function 
\begin{align*}
	\varpi(k;q) = q_0 + q_1 \, \e^{- \ii k_1} + q_2 \, \e^{- \ii k_2}
\end{align*}
and the short-hand notation $q=(q_0,q_1,q_2,q_3)\in\C^3\times\R$, $k=(k_1,k_2)\in\T^2$. 
For fixed $(k;q)$, the eigenvalues of $H(k;q)$ are $E_{\pm}(k;q)=\pm \sqrt{q_3^2 + \sabs{\varpi(q;k)}^2}$, and hence the total spectrum 
\begin{align*}
	\sigma \bigl ( H(q) \bigr ) := \bigcup_{k\in\T^2} \sigma \bigl ( H(k;q) \bigr ) 
	= I_-(q) \cup I_+(q) 
\end{align*}
consists of two bands $I_{\pm}(q):= \ran E_{\pm}(\cdot;q)$. Clearly, $\sigma \bigl ( H(q) \bigr )$ is symmetric around $E=0$ and has a gap there if and only if $E_{+}(k;q) > 0$ holds for all $k \in \T^2$. This allows us to define the gapped parameter space $\bigl ( \C^3 \times \R \bigr )_0 = \bigl ( \C^3 \times \R \bigr ) \setminus Q_{\mathrm{ng}}$ at Fermi energy $E=0$ where 
\begin{proposition}[\cite{Hasegawa_Konno_Nakano_Kohomoto:zero_modes_honeycomb:2006}]
	\begin{align*}
		Q_{\mathrm{ng}} &:= \Bigl \{ (0,q_1,q_2,0) \; \; \big \vert \; \; \sabs{q_1} = \sabs{q_2} \Bigr \} 
		\, \cup
		\\
		&\qquad \cup 
		\Bigl \{ (q_0,q_1,q_2,0) \; \; \big \vert \; \; q_0 \neq 0 , \; \; \babs{\abs{\nicefrac{q_1}{q_0}} - \abs{\nicefrac{q_2}{q_0}}} \leqslant 1 \leqslant \babs{\abs{\nicefrac{q_1}{q_0}} + \abs{\nicefrac{q_2}{q_0}}} \Bigr \} 
		. 
	\end{align*}
\end{proposition}
\begin{proof}
	The proof is straightforward: from the form of the eigenvalues, $E_{\pm}(k;q)$, it is clear that $q_3 = 0$ is a necessary condition for there to be no gap. So let $q_3 = 0$. Then $\sabs{\varpi(q;k)} = 0$ is equivalent to $\varpi(q;k) = 0$. In case $q_0 = 0$, then $\sabs{q_1} = \sabs{q_2}$ is a necessary, and indeed sufficient condition (the points of intersection depend on the phase of the complex numbers $q_1$ and $q_2$). In the remaining case, $q_0 \neq 0$ can be factored out and we can rewrite $\varpi(q;k) = 0$ as 
	\begin{align*}
		1 = - \tfrac{q_1}{q_0} \, \e^{- \ii k_1} - \tfrac{q_2}{q_0} \, \e^{- \ii k_2} 
		, 
	\end{align*}
	and taking the square of the modulus of both sides yields 
	\begin{align*}
		1 = \abs{\frac{q_1}{q_0}}^2 + \abs{\frac{q_2}{q_0}}^2 + 2 \, \Re \left( \frac{\overline{q_1} \, q_2}{q_0^2} \, \e^{+ \ii (k_1 - k_2)} \right) 
		. 
	\end{align*}
	The right-hand side takes values in $\bigl [ (\rho_1 - \rho_2)^2 , (\rho_1 + \rho_2)^2 \bigr ]$, $\rho_j := \babs{\nicefrac{q_j}{q_0}}$, and thus the above equation has a solution if and only if $q$ is in the second region of $Q_{\mathrm{ng}}$. 
\end{proof}
The gapped parameter space is independent of phase and sign of the hopping parameters. Moreover, as the model is trivial if two of the three hopping parameters vanish, we can set $q_0=1$ by choosing a suitable energy scale for the hamiltonian \eqref{nn_model:eqn:parametric_hamiltonian}. Hence, without loss of generality, we can restrict our attention to the parameter space
\begin{align}\label{eq:paraspaces_setting}
	Q = \{ 1 \} \times [0,2] \times [0,2] \times[-1,1] \subset \C^3 \times \R
	.
\end{align}
When necessary, we will smoothen the corners of this cube so that $Q$ is a $\Cont^{\infty}$-manifold. (But even with sharp corners $\pi_1^1(Q) \simeq \pi_1(Q)$ holds which is all we really need.) 
We have sketched the parameter space of gapped configurations $Q_0=Q\setminus Q_{\mathrm{ng}}$ in Figure~\ref{intro:figure:spectrum_H_nn}~(b). Then one immediately deduces the following result:
\begin{proposition}[Fundamental group of $Q_0$]\label{nn_model:prop:pi_1_Qgap}
	The fundamental group $\pi_1(\Qgap) \simeq \Z^2$ is generated by 
	\begin{align}
		\eta_1(t) &= \bigl ( 1 + \eps \, \cos t , 0 , - \eps \, \sin t)
		, 
		\quad \quad 
		\eta_2(t) = \bigl ( 0 , 1 + \eps \, \cos t , - \eps \, \sin t)
		, 
		\label{nn_model:eqn:generators_pi_1}
	\end{align}
	for some $\eps$ small enough. 
\end{proposition}
A detailed proof of this quite obvious fact would be unnecessarily technical, so let us just sketch the basic idea using Figure~\ref{intro:figure:spectrum_H_nn}~(b): It is clear that $\eta_1$ and $\eta_2$ each generate distinct fundamental loops. One may think that loops which wrap around the gapless region and run from $q_I= \bigl ( \nicefrac{3}{2},\nicefrac{1}{4},0 \bigr )$ to $q_{II} = \bigl ( \nicefrac{1}{4},\nicefrac{3}{2},0 \bigr )$ and back represent yet another distinct class of loops. However, these loops can be deformed into a composition of loops equivalent to $\eta_2-\eta_1$. 

\subsection{The Chern numbers} 
\label{nn_model:Chern_numbers}
On $Q$, the three prefactors \eqref{nn_model:eqn:parametric_hamiltonian} of the Pauli matrices simplify to 
\begin{align}
	h_1(k,q) &:= \Re \bigl ( \varpi(q;k) \bigr ) 
	= 1 + q_1 \, \cos k_1 + q_2 \, \cos k_2 
	\notag \\
	h_2(k;q) &:= \Im \bigl ( \varpi(q;k) \bigr ) 
	= q_1 \, \sin k_1 + q_2 \, \sin k_2 
	\label{nn_model:eqn:parameter_hamiltonian_real}
	\\
	h_3(k;q) &= q_3 
	\notag 
\end{align}
The symmetry of the model manifests itself in 
\begin{align*}
	h_1 \bigl ( k_1 ,\cancel{ k_2}; \eta_1(t) \bigr ) &= h_2 \bigl ( \cancel{k_1} , k_2 ; \eta_2(t) \bigr )
\end{align*}
where the $\eta_j$ are the generators of the fundamental groups from Proposition~\ref{nn_model:prop:pi_1_Qgap}; Hence, it suffices to calculate the Chern numbers for one of the generators. The main purpose of this section is to prove the following 
\begin{theorem}\label{nn_model:thm:Chern_numbers}
	Let $\eta \in \Cont(\Sone,Q_0)$ be a loop and $[\eta] = n_1 \, [\eta_1] + n_2 \, [\eta_2] \in \pi_1(Q_0)$ its equivalence class where the $\eta_j$ are given by \eqref{nn_model:eqn:generators_pi_1}. Then the matrix of Chern numbers (\cf Section~\ref{KSV:Bloch_bundle}) is given 
	by 
	\begin{align*}
		C([\eta]) &= n_1 \, C ([\eta_1]) + n_2 \, C ([\eta_2])
		\\
		&= n_1 \, \left (
		\begin{array}{c c : c}
			0 & 0 & 1 \\
			0 & 0 & 0 \\ \hdashline
			-1 & 0 & 0 \\
		\end{array}
		\right )
		+ n_2 \, \left (
		\begin{array}{c c : c}
			0 & 0 & 0 \\
			0 & 0 & 1 \\ \hdashline
			0 & -1 & 0 \\
		\end{array}
		\right )
		. 
	\end{align*}
\end{theorem}
We follow the ideas put forth in Section~\ref{topology_Bloch:m_equal_1} in order to compute $C_1([\eta_1])$ component-by-component. In view of \eqref{nn_model:eqn:parameter_hamiltonian_real} and \eqref{topology_Bloch:eqn:angle_variable02s}, the angle coordinates are 
\begin{align*}
	\theta_{\eta_1}(k_1,\cancel{k_2},t) &= \arccos\left( \frac{- \eps \, \sin t}{\sqrt{1 + 2(1 + \eps \, \cos t) \, \cos k_1+ (1 + \eps \, \cos t)^2 + \eps^2 \, \sin^2 t}}\right)
	, 
	\\
	\varphi_{\eta_1}(k_1,\cancel{k_2},t) &= \arctan\left( \frac{ (1 + \eps \, \cos t) \, \sin k_1}{1 + (1 + \eps \, \cos t) \, \cos k_1}\right)
	. 
\end{align*}
The north and south pole varieties are determined by $\theta_{\eta}$: for $\Nvar_{\eta_1}$ and $\Svar_{\eta_1}$, we need to solve 
\begin{align*}
	 \eps \, \sin t = \mp 
	 \sqrt{ 1 + 2(1 + \eps \, \cos t) \, \cos k_1+ (1 + \eps \, \cos t)^2 + \eps^2 \, \sin^2 t}
	 , 
\end{align*}
and a short computation yields 
\begin{align*}
	\Nvar_{\eta_1} &= \left \{ \bigl ( -\pi , k_2 , -\tfrac{\pi}{2} \bigr ) \; \; \big \vert \; \; k_2 \in [-\pi,+\pi) \right \} 
	, 
	\\
	\Svar_{\eta_1} &= \left \{ \bigl ( -\pi , k_2 , +\tfrac{\pi}{2} \bigr ) \; \; \big \vert \; \; k_2 \in [-\pi,+\pi) \right \} 
	.
\end{align*}
Our first task is to verify whether the intersections of the tori 
\begin{align*}
	\T^2_{1,3} &= \left \{ (k_1,k_2,t) \in \T^3 \; \; \big \vert \; \; k_2 = \ast \in [-\pi,+\pi) \right \} 
	\\
	\T^2_{2,3} &= \left \{ (k_1,k_2,t) \in \T^3 \; \; \big \vert \; \; k_1 = \ast \in [-\pi,+\pi) \right \} 
	\\
	\T^2_{1,2} &= \left \{ (k_1,k_2,t) \in \T^3 \; \; \big \vert \; \; t = \ast \in [-\pi,+\pi) \right \} 
\end{align*}
with the north and south pole varieties are empty, because if at least one of them is, the Chern number has to be $0$. The intersections with the first torus consists only of a single point, 
\begin{align*}
	\Nvar_{\eta_1} \cap \T^2_{1,3} &= \left \{ \bigl ( -\pi , \ast , -\tfrac{\pi}{2} \bigr ) \right \}
	, 
	\qquad \qquad 
	\Svar_{\eta_1} \cap \T^2_{1,3} = \left \{ \bigl ( -\pi , \ast , \tfrac{\pi}{2} \bigr ) \right \} 
	. 
\end{align*}
Whether the intersections with $\T^2_{2,3}$ are empty or $\Nvar_{\eta_1}$ and $\Svar_{\eta_1}$ themselves depends on the choice of embedding point: 
\begin{align*}
	\Nvar_{\eta_1} \cap \T^2_{2,3} &= 
	\begin{cases}
		\emptyset & \ast \neq -\pi \\
		\Nvar_{\eta_1} & \ast = -\pi \\
	\end{cases}
	, 
	\qquad \qquad 
	\Svar_{\eta_1} \cap \T^2_{2,3} = 
	\begin{cases}
		\emptyset & \ast \neq -\pi \\
		\Svar_{\eta_1} & \ast = -\pi \\
	\end{cases}
\end{align*}
From a topological perspective, the embedding point (\ie the choice of $k_2 = \ast$) is not important, so we conclude $C_{23}([\eta_1]) = 0$. For the sake of completeness, we will postpone the explicit computation for the case $\ast = \pi$ to the end of the section. 
In case of the last torus, at least one of the intersections is always empty
independently of the choice of $\ast$, and thus $C_{1,2}([\eta_1]) = 0$.

To compute $C_{1,3}([\eta_1])$, we pick an open set $\mathcal{K}_{\Nvar,\eta_1}$ which contains $p_{\Nvar} = \bigl ( -\pi,\ast,-\tfrac{\pi}{2} \bigr )$ so that the boundary $\partial \mathcal{K}_{\Nvar,\eta_1}$ is parametrized by the circle 
\begin{align*}
	\gamma_{1,3}(s) = p_{\Nvar} + \delta \, \bigl ( \cos s , 0 , \sin s \bigr ) 
\end{align*}
for some $\delta > 0$ small enough. The exterior derivative 
\begin{align*}
		\dd \varphi_{\eta_1}(k_1,\cancel{k_2},t)&= \partial_{k_1} \varphi_{\eta_1}(k_1,t) \, \dd k_1 + \partial_t \varphi_{\eta_1}(k_1,t) \, \dd t 
\end{align*}
depends on 
\begin{align*}
	\partial_{k_1} \varphi_{\eta_1}(k_1,t) &= - \frac{(1 + \eps \, \cos t) (1 + \eps \, \cos t + \cos k_1)}{1 + 2 \, (1 + \eps \, \cos t) \, \cos k_1 + (1 + \eps \, \cos t)^2}
	, 
	\\
	\partial_t \varphi_{\eta_1}(k_1,t) &= \frac{\eps \, \sin k_1 \, \sin t}{1 + 2 \, (1 + \eps \, \cos t) \, \cos k_1 + (1 + \eps \, \cos t)^2}
	. 
\end{align*}
Expanding the partial derivatives evaluated at $\gamma_{1,3}(s)$ for small $\delta$, we obtain 
\begin{align*}
	\partial_{k_1} \varphi_{\eta_1} \bigl ( \gamma_{01}(s) \bigr ) &= - \frac{- \delta \, \eps \, \sin s + \order(\delta^2)}{\delta^2 \bigl ( \cos^2 s + \eps^2 \, \sin^2 s \bigr ) + \order(\delta^3)} 
	= - \frac{1}{\delta} \frac{\eps \, \sin s}{\cos^2 s + \eps^2 \, \sin^2 s} 
	+ \order(1) 
	, 
	\\
	\partial_t \varphi_{\eta_1} \bigl ( \gamma_{01}(s) \bigr ) &= \frac{- \delta \, \eps \, \cos s + \order(\delta^2)}{\delta^2 \bigl ( \cos^2 s + \eps^2 \, \sin^2 s \bigr ) + \order(\delta^3)} 
	= \frac{1}{\delta} \frac{\eps \, \cos s}{\cos^2 s + \eps^2 \, \sin^2 s} 
	+ \order(1) 
	. 
\end{align*}
Now we plug in the parametrization into equation~\eqref{topology_Bloch:eqn:Chern_number_integral_formula} to obtain 
\begin{align*}
	C_{1,3}([\eta_1]) &= \frac{1}{2\pi} \ointclockwise_{\partial (\T^2_{01} \cap K_{\Nvar,\eta_1})} \dd \varphi_{\eta_1} 
	= \frac{1}{2\pi} \int_0^{2\pi} \dd s \, \nabla \varphi_{\eta_1} \bigl ( \gamma_{1,3}(s) \bigr ) \cdot \dot{\gamma}_{13}(s) 
	\\
	&= \frac{\eps}{2\pi} \int_0^{2\pi} \frac{\dd s}{\cos^2 s + \eps^2 \, \sin^2 s} + \order(\delta) 
	= 1 
	. 
\end{align*}
To arrive at the last equality, we note that the leading-order integral can be computed explicitly and that Chern numbers need to be integers. 

For the benefit of the reader, we will double-check that $C_{2,3}([\eta_1]) = 0$ explicitly, even if we choose the “bad” embedding point $\ast = -\pi$. Here, the singular string $\Nvar_{\eta_1} \subset \T^2_{2,3}$ is fully contained in the torus. Pick an open set $\mathcal{K}_{\Nvar,\eta_1} \supset \Nvar_{\eta_1}$ so that $\partial\mathcal{K}_{\Nvar,\eta_1}$ can be parametrized by 
\begin{align*}
	\gamma_{2,3}^{\pm}(s) := \bigl ( -\pi , s , -\tfrac{\pi}{2} \pm \delta \bigr ) 
\end{align*}
for $\delta > 0$ small enough. Then from $\dot{\gamma}^{\pm}_{2,3}(s) = (0,1,0)$ and $\partial_{k_2} \varphi_{\eta_1} \bigl ( \gamma^{\pm}_{2,3}(s) \bigr ) = 0$, we conclude that 
$$
\nabla \varphi_{\eta_1} \bigl ( \gamma_{2,3}^{\pm}(s) \bigr ) \cdot \dot{\gamma}^{\pm}_{2,3}(s) 
	= 0 
$$
which implies that $C_{2,3} ([\eta_1])$ vanishes. 

\subsection{The effect of perturbations} 
\label{nn_model:perturbation}
We are interested in discussing the effect of a small perturbation $\lambda V$ (periodic or random) to the operator \eqref{nn_model:eqn:parametric_hamiltonian}. As in Section \ref{intro:randomness} we will assume $\lambda\ll1$ and $\|V\|_{\Alg}=1$.

The assumption that there exists a global gap of size at least $g$ for the hamiltonian \eqref{nn_model:eqn:parametric_hamiltonian} is equivalent to 
$$
q_3^2 + \sabs{\omega(q;k)}^2>\frac{g^2}{4},\qquad\quad\forall\; k\in\T^2.
$$
This condition is of course realized if $|q_3|> \frac{g}{2}$. For $(q_1,q_2,0) \in Q_{0}$, the above condition is realized if one of the following inequalities is verified
$$
q_2<q_1-1-\frac{g}{2},\qquad q_2>q_1+1+\frac{g}{2}, \qquad q_2<-q_1+1-\frac{g}{2}.
$$
The space $Q_{0,g}\subset Q_{0}$ obtained by the configurations which verify the above relations is represented in Figure \ref{intro:figure:random_parameter_space}. It is clear that $Q_{0,g} \rightarrow Q_{0}$ when $g \rightarrow 0$, \ie $Q_{0,g}$ \emph{deformation retracts} to the thin set $Q_{0}$ in a topological sense \cite{Hatcher:algebraic_topology:2002}. This means that $Q_{0,g}$ is homotopically equivalent to $Q_0$ (at least for $g<2$) and so $\pi_1(Q_{0,g})=\pi_1(Q_{0})$.

Any $\eta\in Q_{0,g}$ can also be considered as a loop in $Q_0$ and its homotopy class can be expressed as $\eta=n_1[\eta_1]+n_2[\eta_2]$. Moreover, for all $\lambda\in[0,\lambda_\ast]$ with $\lambda_{\ast} < \frac{g}{2}$ we can compute the Chern class $C_{j,n,\lambda} ([\eta])$ for the perturbed hamiltonian $H_{\lambda,\eta}(t)$. The homotopy equivalence of the spectral projections (Proposition~\ref{prop:homotopy_disorder}) implies $C_{j,n,\lambda} ([\eta]) = C_{j,n} ([\eta])$ agrees with the Chern numbers from the periodic case.

%
\begin{appendix}
	\section{Symmetric classes for two-band periodic systems} 
\label{two-bands-symmetries}

Let us consider a  \emph{two-band periodic operator} in the  algebra $\mathfrak{A}_{\text{per}}=\mathfrak{S}\otimes\text{Mat}_{2^m}(\C)$ of the form
\begin{align}\label{eq:sym:two_level}
	\hat{H} = \sum_{j = 1}^{2m+1} \mathfrak{h}_j \otimes \Sigma_j 
\end{align}
where $\mathfrak{h}_j=\mathfrak{h}_j^\ast\in \mathfrak{S}$.
In definition \eqref{eq:sym:two_level} there is no loss of generality compared to \eqref{eq:two_level} since $\hat{H}$ and $\hat{H} - \mathfrak{h}_0 \otimes \id_{2^m}$ have the same spectral projections. 

An explicit realization for the  Clifford matrices $\Sigma_j$ is given by
\begin{equation}\label{eq:clifford_matrix}
\begin{aligned}
&\Sigma_1&:=&\quad \sigma_1\otimes\sigma_3\otimes\sigma_3\otimes\ldots\otimes\sigma_3\\
&\Sigma_2&:=&\quad \sigma_2\otimes\sigma_3\otimes\sigma_3\otimes\ldots\otimes\sigma_3\\
&\Sigma_3&:=&\quad \id_2\otimes\sigma_1\otimes\sigma_3\otimes\ldots\otimes\sigma_3\\
&\Sigma_4&:=&\quad \id_2\otimes\sigma_2\otimes\sigma_3\otimes\ldots\otimes\sigma_3\\
&\ \vdots &&\ \ \ \ \ \ \ \ \ \ \ \ \ \ \ \ \ \vdots\ \\
&\Sigma_{2m-1}&:=&\quad \id_2\otimes\id_2\otimes\id_2\otimes\ldots\otimes\sigma_1\\
&\Sigma_{2m}&:=&\quad \id_2\otimes\id_2\otimes\id_2\otimes\ldots\otimes\sigma_2\\
&\Sigma_{2m+1}&:=&\quad \sigma_3\otimes\sigma_3\otimes\sigma_3\otimes\ldots\otimes\sigma_3.
\end{aligned}
\end{equation}
where each term is a tensor products of $m$ Pauli matrices $\sigma_1,\sigma_2,\sigma_3$. This choice is essentially unique, up to unitary equivalences \cite{Lee:1948}. Since $\sigma_1,\sigma_3$ are real matrices and $\sigma_2$ is purely imaginary it follows that $\overline{\Sigma}_j=(-1)^{j+1}\ \Sigma_j$.

The peculiarity of these models resides in the existence of a symmetry which is reflected in the spectral properties. First of all, let us recall the \emph{complex conjugation}  which acts antilinearly as  $C:\Psi\mapsto\overline{\Psi}$ for $\Psi\in \ell^2(\Gamma)\otimes\C^r$ ($r$ is arbitrary). Moreover, we will also need the \emph{parity operator} $\wp$ on $\ell^2(\Gamma)$ which in particular acts by conjugation on selfdajoint $\mathfrak{h}\in\mathfrak{S}$ as
$\wp\, \mathfrak{h}\,\wp=C\,\mathfrak{h}\,C$.
Moreover $\wp=\wp^\ast$ and $\wp^2=\id_{\ell^2(\Gamma)}$.

Let us start with the \emph{odd} case $m=2\nu-1$ and introduce the matrix 
\begin{equation}\label{eq:symm_odd}
\Theta:=\Sigma_1\Sigma_3\ldots\Sigma_{4\nu-1},\qquad\quad \Theta^\ast=\Theta^{-1},\qquad\quad\Theta^2=(-1)^{\nu}\,\id_{2^m} 
. 
\end{equation}
Since $\Theta$ is a product of $2\nu$ matrices (with odd indices), exploiting
the Clifford algebra relations one obtains
$\Theta\,\Sigma_j=(-1)^{j}\,\Sigma_j\,\Theta$. If one introduces the unitary operator $U_\Theta:= \wp\otimes \Theta$, with involution property $U_\Theta^2=(-1)^\nu\,\id_{\mathcal{H}}$, one can verify that
\begin{equation}\label{symm_two_lewel}
U_\Theta\; \hat{H}\;U_\Theta^\ast=-\,C\; \hat{H} \;C
. 
\end{equation}
The relation \eqref{symm_two_lewel} says that the operator $\hat{H}$
has a \emph{particle-hole} (PH) symmetry in the language of
\cite{Schnyder_Ryu_Furusaki_Ludwig:classification_topological_insulators:2008}

For the even case $m=2\nu$ one introduces the matrix 
\begin{equation}\label{eq:symm_even}
\Upsilon:=\Theta\Sigma_{4\nu+1},\qquad\quad \Upsilon^\ast=\Upsilon^{-1},\qquad\quad\Upsilon^2=(-1)^{\nu}\,\id_{2^m}
\end{equation}
Since $\Upsilon$ is a product of $2\nu + 1$ matrices (with odd indices), one obtains $\Upsilon\,\Sigma_j=(-1)^{j+1}\,\Sigma_j\,\Upsilon$. The unitary operator $U_\Upsilon:= \wp\otimes \Upsilon$, with involution property $U_\Upsilon^2=(-1)^\nu\,\id_{\mathcal{H}}$, provides
\begin{equation}\label{symm_two_lewel_even}
U_\Upsilon\; \hat{H}\;U_\Upsilon^\ast=\,C\; \hat{H} \;C
. 
\end{equation}
The relation \eqref{symm_two_lewel_even} says that the operator $\hat{H}$
has a \emph{time-reversal} (TR) symmetry again according to
\cite{Schnyder_Ryu_Furusaki_Ludwig:classification_topological_insulators:2008}

One has the following table of symmetries for $\hat{H}$ depending on the spinorial dimension: 

\begin{table}[htdp]\label{tab:anti-symm}
\begin{center}
\begin{tabular}{|c|c|c|c|c|}
\hline
$m$ & $\nu$ & {\bf PH} & {\bf TR} & {\bf AZ} \\
\hline
\hline
even & even & $0$ & $+1$ & AI\\
\hline
even & odd & $0$ & $-1$ & AII\\
\hline
odd & even & $+1$ & $0$ & D\\
\hline
odd& odd & $-1$ & $0$ & C\\
\hline
\end{tabular}
\end{center}
\caption{
Table of the discrete symmetries for two-band periodic hamiltonian $\hat{H}$ of type \eqref{eq:sym:two_level}. The classification depend on the dimension of the internal degrees of freedom through the parity of $m$ and $\nu:= \lfloor (m+1)/2 \rfloor$. The sign of the PH-symmetry (resp. TR-symmetry) is the sign of $U_\Theta^2$ (resp. $U_\Upsilon^2$). The $0$ means absence of symmetry. The column AZ provides the Cartan's label according to the Altland-Zirnbauer classification \cite{Altland_Zirnbauer:superconductors_symmetries:1997, Schnyder_Ryu_Furusaki_Ludwig:classification_topological_insulators:2008}. }
\end{table}
\end{appendix}

\printbibliography

\end{document}